\documentclass[lettersize,journal]{IEEEtran}
\usepackage{amsmath,amsfonts}
\usepackage{amsthm,amssymb}
\usepackage{mathrsfs}
\usepackage{color}
\usepackage{algorithmic}
\usepackage{algorithm}
\usepackage{array}
\usepackage[caption=false,font=footnotesize,labelfont=rm,textfont=rm]{subfig}
\usepackage{textcomp}
\usepackage{stfloats}
\usepackage{url}
\usepackage{verbatim}
\usepackage{graphicx}
\usepackage{float}
\usepackage{cite}
\usepackage{hyperref}
\usepackage{orcidlink}
\hypersetup{hidelinks}
\newcommand{\Ab}{\mathbf{A}}

\newcommand{\fb}{\mathbf{f}}
\newcommand{\hb}{\mathbf{h}}
\newcommand{\nb}{\mathbf{n}}
\newcommand{\rb}{\mathbf{r}}

\newcommand{\ub}{\mathbf{u}}
\newcommand{\vb}{\mathbf{v}}
\newcommand{\wb}{\mathbf{w}}
\newcommand{\xb}{\mathbf{x}}
\newcommand{\yb}{\mathbf{y}}

\newcommand{\Hb}{\mathbf{H}}
\newcommand{\Ib}{\mathbf{I}}
\newcommand{\Rb}{\mathbf{R}}
\newcommand{\Ub}{\mathbf{U}}
\newcommand{\Vb}{\mathbf{V}}

\newcommand{\Cbb}{\mathbb{C}}
\newcommand{\Ebb}{\mathbb{E}}
\newcommand{\Rbb}{\mathbb{R}}

\newcommand{\Lcal}{\mathcal{L}}
\newcommand{\Ncal}{\mathcal{N}}
\newcommand{\Ocal}{\mathcal{O}}

\newcommand{\CNcal}{\mathcal{CN}}
\newcommand{\Sigmab}{\boldsymbol{\Sigma}}

\newcommand{\rhob}{\boldsymbol{\rho}}

\newcommand{\dB}{\operatorname{dB}}
\newcommand{\Prob}[1]{\mathbf{P#1}}
\newtheorem{lemma}{Lemma}

\begin{document}

\title{Integrated Sensing-Communication-Computation for Over-the-Air Edge AI Inference}

\author{Zeming~Zhuang$^{\orcidlink{0000-0001-8846-523X}}$,~\IEEEmembership{Student~Member,~IEEE,} Dingzhu~Wen$^{\orcidlink{0000-0003-0538-5811}}$,~\IEEEmembership{Member,~IEEE,} Yuanming~Shi$^{\orcidlink{0000-0002-1418-7465}}$,~\IEEEmembership{Senior~Member,~IEEE,} Guangxu~Zhu$^{\orcidlink{0000-0001-9532-9201}}$,~\IEEEmembership{Member,~IEEE,} Sheng~Wu$^{\orcidlink{0000-0002-9947-9968}}$,~\IEEEmembership{Member,~IEEE,} and~Dusit~Niyato$^{\orcidlink{0000-0002-7442-7416}}$,~\IEEEmembership{Fellow,~IEEE} 

\thanks{Z. Zhuang, D. Wen, and Y. Shi are with Network Intelligence Center, School of Information Science and Technology, ShanghaiTech University, Shanghai, China (e-mail: \{zhuangzm, wendzh, shiym\}@shanghaitech.edu.cn). (Corresponding authors: D. Wen and Y. Shi).}
\thanks{G. Zhu is with Shenzhen Research Institute of Big Data, Shenzhen, China (e-mail: gxzhu@sribd.cn).}
\thanks{S. Wu is with the School of Information and Communication Engineering, Beijing University of Posts and Telecommunications, Beijing 100876, China (e-mail: thuraya@bupt.edu.cn).}
\thanks{D. Niyato is with School of Computer Science and Engineering, Nanyang Technological University, Singapore (e-mail: dniyato@ntu.edu.sg).}
}



\maketitle

\begin{abstract}
    Edge-device co-inference refers to deploying well-trained artificial intelligent (AI) models at the network edge under the cooperation of devices and edge servers for providing ambient intelligent services. For enhancing the utilization of limited network resources in edge-device co-inference tasks from a systematic view, we propose a task-oriented scheme of integrated sensing, computation and communication (ISCC) in this work. In this system, all devices sense a target from the same wide view to obtain homogeneous noise-corrupted sensory data, from which the local feature vectors are extracted. All local feature vectors are aggregated at the server using over-the-air computation (AirComp) in a broadband channel with the orthogonal-frequency-division-multiplexing technique for suppressing the sensing and channel noise. The aggregated denoised global feature vector is further input to a server-side AI model for completing the downstream inference task. A novel task-oriented design criterion, called maximum minimum pair-wise discriminant gain, is adopted for classification tasks. It extends the distance of the closest class pair in the feature space, leading to a balanced and enhanced inference accuracy. Under this criterion, a problem of joint sensing power assignment, transmit precoding and receive beamforming is formulated. The challenge lies in three aspects: the coupling between sensing and AirComp, the joint optimization of all feature dimensions' AirComp aggregation over a broadband channel, and the complicated form of the maximum minimum pair-wise discriminant gain. To solve this problem, a task-oriented ISCC scheme with AirComp is proposed. Experiments based on a human motion recognition task are conducted to verify the advantages of the proposed scheme over the existing scheme and a baseline.
\end{abstract}



\section{Introduction}
The next generation of wireless technology (6G) will go far beyond just communication services to push forward an era of true Intelligence of Everything (IoE) for providing immersive intelligent services like auto-driving, Metaverse, smart city, etc. \cite{letaief2019the, zhu2020toward, nguyen20226g, shi2020communication, xu2023unleashing, wang2022federated}. However, the realization of these services highly depends on utilizing the inference capability of well-trained AI models at the network edge for intelligent decision making. This gives rise to a new research topic called edge AI inference, or edge inference \cite{liu2023resource, shao2020comm, letaief2022edge,shi2023taskoriented}. 

The implementation of edge inference includes three paradigms, i.e., on-device inference, on-server inference and edge-device co-inference. In on-device inference, well-trained AI models are downloaded by edge devices for executing inference tasks, leading to heavy computation overhead (see, \cite{lee2023decent,cai2020onceforall,lu2022improving}). To alleviate the computation bottleneck at devices, the on-server inference uploads the raw data samples from devices to an edge server, where large-scale AI models are deployed for inference (see, \cite{yang2020energy, hua2021reconf, yang2020sparse}). This, however, violates the data privacy of edge devices. To further address the privacy issue, the edge-device co-inference emerges as a promising solution (see, \cite{shi2019improving, shao2021branchy, niu2022an, yan2022optimal}). It divides an AI model into two parts. The front-end part has a smaller size and is deployed at devices for feature extraction. The computation-intensive back-end part is deployed at the server, which leverages the received local feature vectors to complete the remaining inference task. As a result, computation is offloaded to the edge server and the avoidance of raw data transmission keeps devices' data privacy. Hence, the edge-device co-inference paradigm is adopted in this work.

Recently, the edge-device co-inference has experienced a rapid advancement. The first research focus is to balance the trade-off between communication and computation. In \cite{shi2019improving, shabbeer2022simple}, the neural network was pruned at training phase to avoid the huge communication overhead caused by in-layer data amplification phenomenon. A suitable split layer selection method was developed in \cite{zhang2021comm} together with the scheme for encoding/decoding the intermediate feature vector by an automated machine learning (AutoML) framework. Besides, methods of setting early exiting points in neural networks were proposed in \cite{liu2023resource, li2020edge, wang2020dual} to balance the communication and computation overhead under a given empirical inference accuracy threshold. The authors in \cite{niu2022an} further combined the methods of early exiting, model partitioning and data quantization to improve the inference performance. A joint source and channel coding (JSCC) approach was developed in \cite{jankowski2021wireless} to map feature vectors into channel symbols. Nevertheless, as stated by \cite{shao2022learning,shao2023task,wen2022task2,zhu2023pushing}, edge inference features a task-oriented property where the effectiveness and efficiency of the inference task execution are of crucial significance. As a result, the conventional design criteria including communication capacity or signal-to-noise ratio (SNR) of received signals work no longer well, as they cannot differentiate the feature elements with the same size and distortion level but different contributions on inference accuracy \cite{wen2022task2}. To address this issue, this work proposes to directly use the inference accuracy as the design criterion.

One main challenge of designing task-oriented schemes is that the instantaneous inference accuracy is unknown and has no mathematical model. To address this issue, the authors in \cite{lan2022progressive} proposed an approximate but tractable metric, called discriminant gain. By considering classification tasks and based on the assumption that the feature vector follows a Gaussian mixture distribution with each Gaussian component corresponding to one class, a pair-wise discriminant gain for two arbitrary classes (called a class pair) is defined as the symmetric Kullback-Leibler (KL) divergence of their distributions. With a larger pair-wise discriminant gain, the two classes can be easily differentiated in the feature space, leading to an enhanced achievable inference accuracy. Existing works (see, \cite{lan2022progressive,wen2022task1, wen2022task2}) use the average of all pair-wise discriminant gains as the design objective. This, however, causes an unbalanced inference accuracy of different classes and degrades the overall inference performance. As shown in Fig.~\ref{figgr:DG}(a), under this design goal, one particular class (i.e., Class 1) may be far separated from all other classes (i.e., Classes 2 and 3), which could be very close to each other in the feature space. 
\begin{figure}[!t]
    \centering
    \subfloat[Maximize average discriminant gain]{
        \includegraphics[width=0.8\columnwidth]{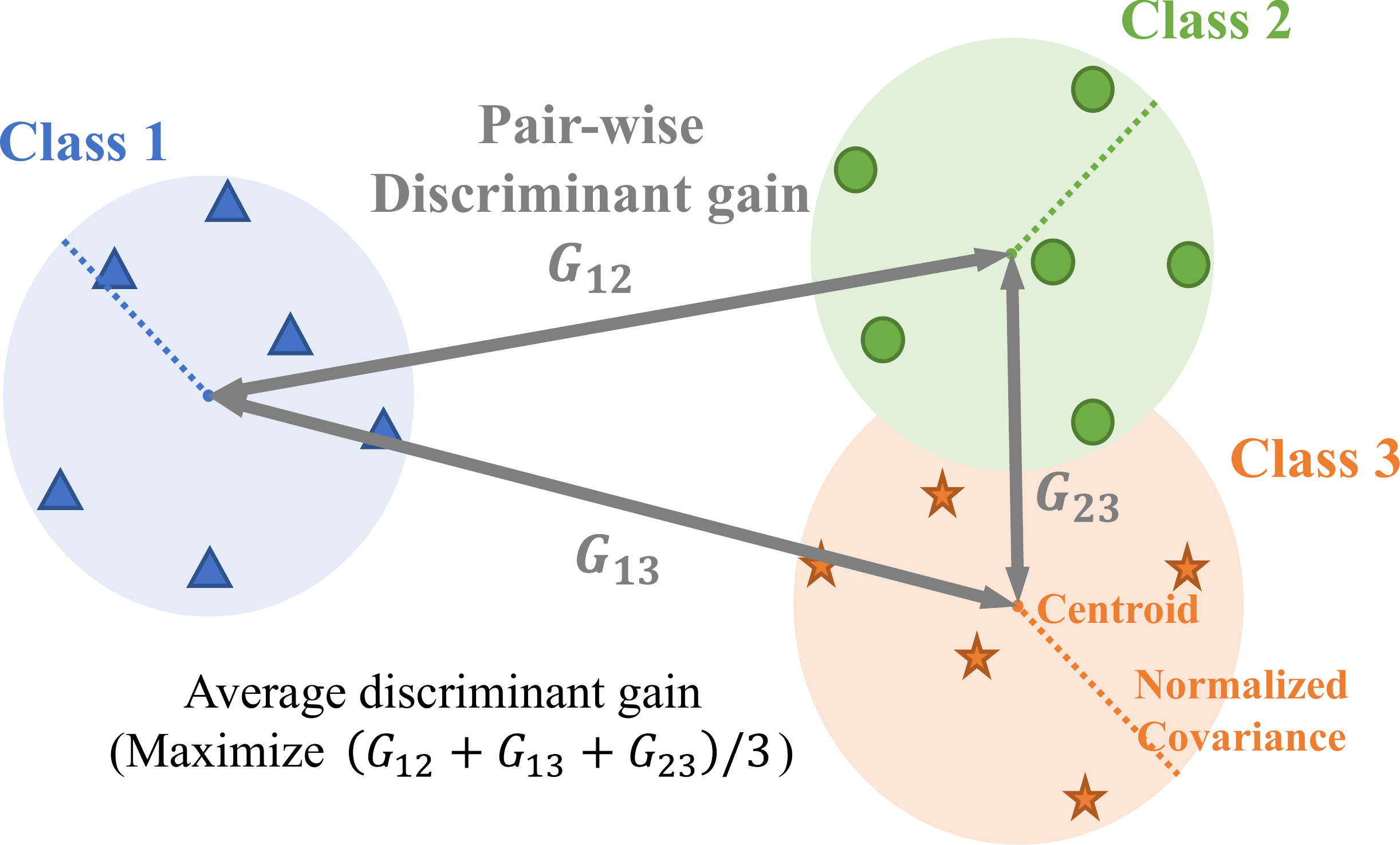}
        \label{fig:avgDG}
    }\\
    \subfloat[Maximize minimum pair-wise discriminant gain]{
        \includegraphics[width=0.8\columnwidth]{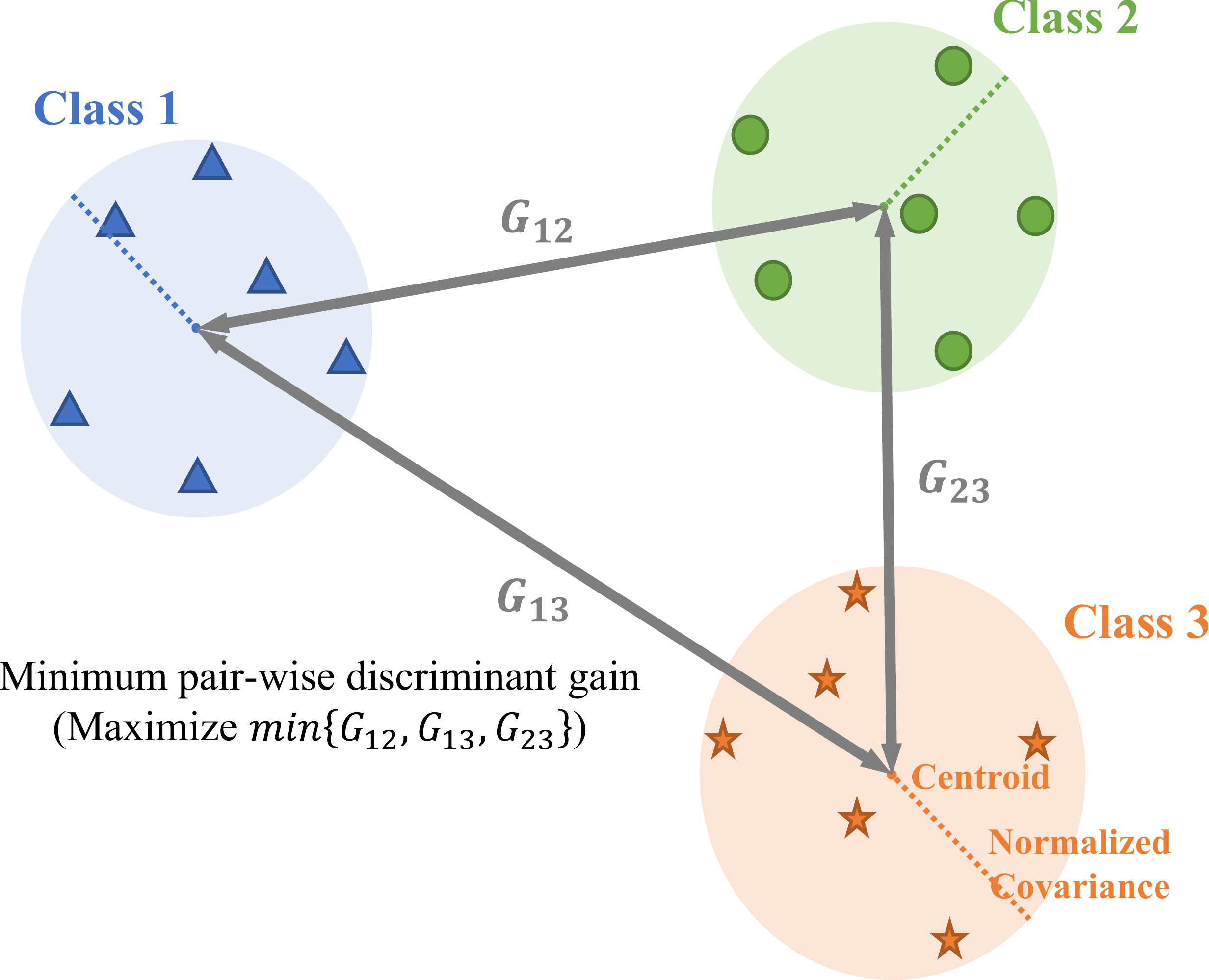}
        \label{fig:minDG}
    }
    \caption{Average discriminant gain maximization v.s. minimum pair-wise discriminant gain maximization.}
    \label{figgr:DG}
\end{figure}
To address this issue, in this work, we target maximizing the minimum pair-wise discriminant gain, which guarantees the closest class pair can be well separated in the feature space, as shown in Fig.~\ref{figgr:DG}(b).


On the other hand, although the previous works can enhance the inference performance, they optimize the edge-device co-inference systems from a partial view (i.e., the perspectives of communication or computation or both), which ignores the influence of the data acquisition process on inference performance and focuses on task offloading, model partitioning or data compressing (see, \cite{hu2023content,lan2022progressive,niu2022an}). Also, many existing works on multi-device ISAC framework have been proposed and developed \cite{wen2023integrated} (e.g., UAV deployment \cite{tang2023integrated}, data redundancy exploitation and sensing-communication switching \cite{li2022multi}). However, they cannot achieve the full potential for enhancing the inference performance. As stated in \cite{wen2022task1}, the fulfillment of an edge-device co-inference task requires the cooperation of sensing for data acquisition, computation for feature extraction and communication for feature transmission, at edge devices. The inference accuracy depends on the feature distortion level caused during the data acquisition, computation and communication three processes. Besides, they compete for network resources including time and energy for suppressing their own distortion. Hence, edge-device co-inference calls for integrated sensing, communication and computation (ISCC) schemes \cite{wen2022task1}. To this end, a task-oriented scheme was proposed in \cite{wen2022task1} for maximizing the inference accuracy. However, the aforementioned work investigates the scenario of narrow-view sensing, which refers to that all devices perceive disjoint small ranges of a source target to obtain high-quality low-dimensional sensory data. There is a lack of ISCC schemes for handling the scenario of wide-view sensing, where each device perceives the same wide range of a source target and acquires noise-corrupted high-dimensional sensory data. To fill this gap, we propose a task-oriented scheme that integrates sensing and over-the-air computation (AirComp) for wide-view sensing based edge-device co-inference systems.

In this paper, a multi-device based ISCC system is considered to support edge-device co-inference tasks in many application scenarios such as ensuring security and reducing energy consumption in smart home (see, \cite{yar2021towards}), autonomous driving (see, \cite{cheng2022integrated}) and traffic monitoring in Vehicle-to-Everything (V2X) (see, \cite{du2023towards}). Each device is equipped with a single antenna and a dual-functional-radar-communication (DFRC) transceiver used both for sensing and communication. First, all devices transmit a frequency modulation continuous wave (FMCW) signal in an orthogonal frequency band to sense the same wide view of the source target for obtaining homogeneous sensory raw data. Then, a singular value decomposition (SVD) based linear filter is adopted for clutter cancellation and a principal component analysis (PCA) based extractor is exploited for extracting a low-dimensional local feature vector at each device. For further suppressing the sensing noise power and enhancing the communication efficiency, all local feature vectors are aggregated at the edge server via the technique of AirComp. Specifically, AirComp allows all devices simultaneously to transmit the same dimension of all local feature vectors over the same frequency band, leading to a significant enhancement of communication efficiency (see, \cite{wang2022over,zhu2018mimo,wen2019reduced,zhu2020one}). By leveraging the waveform superposition property, a weighted sum of all local feature elements is directly calculated instead of decoding the value of each one individually. This work jointly considers the aggregation of all elements over an orthogonal frequency division multiplexing (OFDM) based broadband channel. Based on the novel design criterion called maximum minimum pair-wise discriminant gain, we propose the joint sensing power assignment, transmit precoding and receive beamforming problem. The challenges to solving this problem arise from three aspects: the coupling between sensing and AirComp, the joint optimization of all feature elements and the complicated form of the maximum minimum pair-wise discriminant gain. To address this problem, we propose the task-oriented ISCC scheme with AirComp. The detailed contributions of this work are summarized as follows.
\begin{itemize}
    \item {\bf Novel Design Metric of Maximum Minimum Pair-Wise Discriminant Gain}: To overcome the limitation of unbalanced and low inference accuracy resulting from the existing metric of average pair-wise discriminant gain (see, \cite{lan2022progressive, wen2022task1, wen2022task2}), we adopt a novel design criterion called maximum minimum pair-wise discriminant gain in this work. It maximizes the discriminant gain between the closest class pair. Consequently, the least distinguishable class pair can be well separated in the feature space. This leads to a balanced and enhanced achievable inference accuracy. 
    
    \item {\bf AirComp based ISCC Framework for Edge-Device Co-Inference}: An AirComp based ISCC framework is established to complete edge-device co-inference tasks. The modules of sensing (including sensing waveform design and SVD based clutter cancellation), on-device computation (i.e., PCA based feature extraction) and AirComp (local feature vectors aggregation) are efficiently constructed. Particularly, an OFDM based broadband channel is used for the aggregation of all local feature vectors. Over an arbitrary frequency subcarrier, the same dimension of all local feature vectors is aggregated. The aggregation of different dimensions is over different subcarriers. The influences of each module on the design metric, i.e., minimum pair-wise discriminant gain, are mathematically characterized in closed-form expressions.  
 
    \item {\bf Task-Oriented ISCC Scheme with AirComp}: Under the criterion of maximum minimum pair-wise discriminant gain, we formulate the problem of joint sensing power assignment, transmit precoding and receive beamforming. We then propose the task-oriented ISCC scheme to address this problem, which first conducts variables transformation to derive an equivalent problem with a difference-of-convex (d.c.) form and then solves the d.c. problem based on the typical method of successive convex approximation (SCA) \cite{razaviyayn2014successive}. Compared with the existing AirComp based scheme in \cite{wen2022task2}, where the optimization of different feature elements is separately designed and the sensing stage is not considered, the sensing, on-device computation and AirComp of all feature elements are jointly optimized in our proposed scheme. This provides two extra degrees of freedom to enhance the inference performance. On one hand, the system is optimized from a systematic view that coordinates the design of sensing, computation and communication by fully considering their coupling mechanism and competence in inference tasks. On the other hand, the joint design of all feature dimensions allows adaptive resource allocation among different feature dimensions, i.e., more resources can be assigned to the more important feature dimensions of the inference task. 
    
    \item {\bf Performance Evaluation}:
    Extensive experiments are performed to evaluate our proposed framework and algorithm based on the wireless sensing simulator proposed in \cite{li2021wireless}. A wide-view human motion recognition task is considered with two inference models: a multi-layer perception (MLP) neural network and a support vector machine (SVM) model. To begin with, the inference accuracy is shown to be monotonically increasing with the maximum minimum pair-wise discriminant gain, which verifies the efficiency of the adopted design criterion. Then, the proposed scheme is shown to outperform the state-of-the-art scheme and a baseline scheme. 
\end{itemize}


\section{System Model and Problem Formulation}\label{sec:2}
\subsection{Network Model}
Consider a single network to support edge-device co-inference tasks, as shown in Fig.~\ref{fig:system}. There is one edge server equipped with an $N_r$-antenna access point (AP) and $K$ edge devices, each of which is equipped with a dual-functional-radar-communication (DFRC) system. Many types of radar are used for sensing in different scenarios including pulsed radar, continuous-wave radar, OFDM radar, OTFS radar, FMCW radar, etc \cite{zhang2022enabling}. Pulsed radar and continuous-wave radar are low-efficiency due to the avoidance of self-interference. The OFDM radar and OTFS radar suffer from co-channel interference from the communication systems \cite{carv2020compa,zhang2022enabling}. In the FMCW radar adopted in this paper, a dedicated frequency band is utilized for sensing and the frequency of the sensing signal is modulated as a linear function of time. As a result, there is no co-channel interference and self-interference\cite{carv2020compa,zhang2022enabling}. The workflow to complete an edge inference task is shown in Fig.~\ref{fig:time}. All devices perceive the same wide view of a source target and obtains homogeneous sensory data, from which the local feature vectors are extracted. The dimension of each local feature vector is denoted as $M$. The sensing frequency bands of different devices are orthogonal. Then, all local feature vectors are aggregated to derive a denoised global feature vector at the edge server using the technique of AirComp. Finally, the global feature vector is input into a server-side AI model to complete the whole inference task. 
\begin{figure*}[!t]
    \centering
    \includegraphics[width=0.9\linewidth]{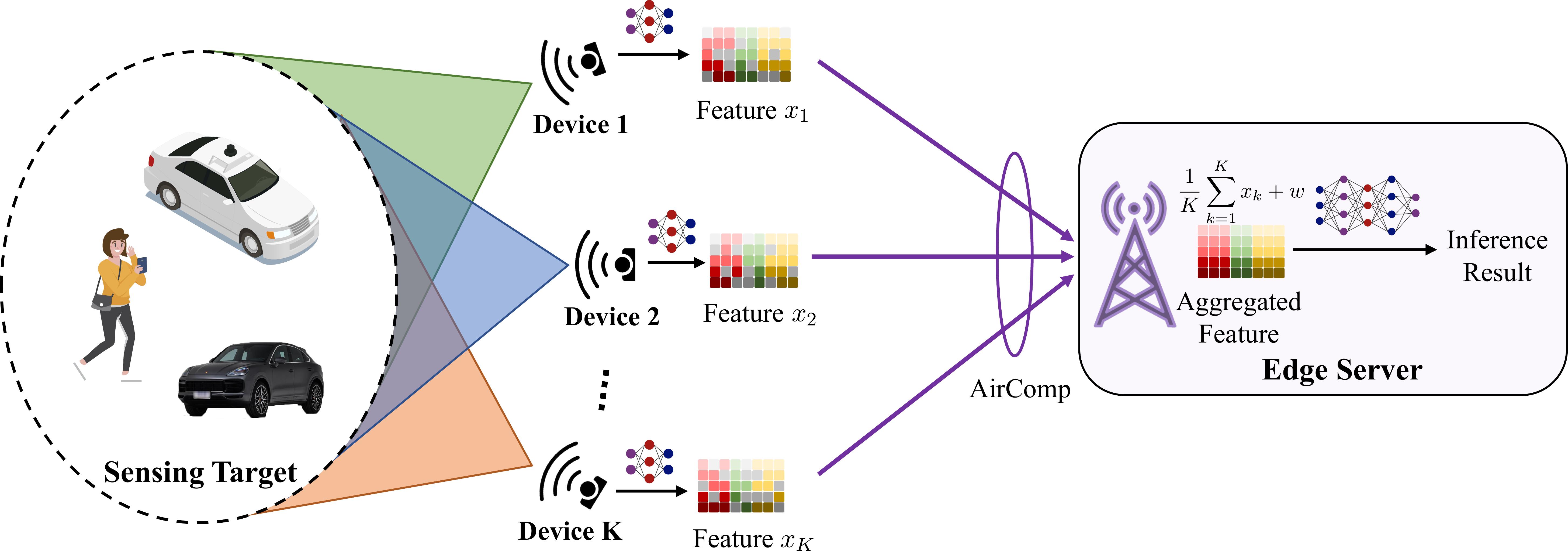}
    \caption{The system architecture of proposed ISCC framework.}
    \label{fig:system}
\end{figure*}

The sensing, computation and AirComp processes operate sequentially at all devices, as shown in Fig.~\ref{fig:time}. Particularly, to aggregate all feature elements using AirComp, OFDM is leveraged. $M$ frequency subcarriers are used to aggregate all the $M$ dimensions of the local feature vectors. Over each subcarrier, an element of the same feature dimension is transmitted by all devices and is aggregated at the edge server to get a global denoised one. As the time length of transmitting one feature element is much shorter than the channel coherence-time duration \cite{Zhu2020TWC}, static channels are assumed during one time slot. The edge server serves as a central coordinator and has the ability to acquire the channel state information (CSI) of all involved links. 
\begin{figure*}[!t]
    \centering
    \includegraphics[width=\linewidth]{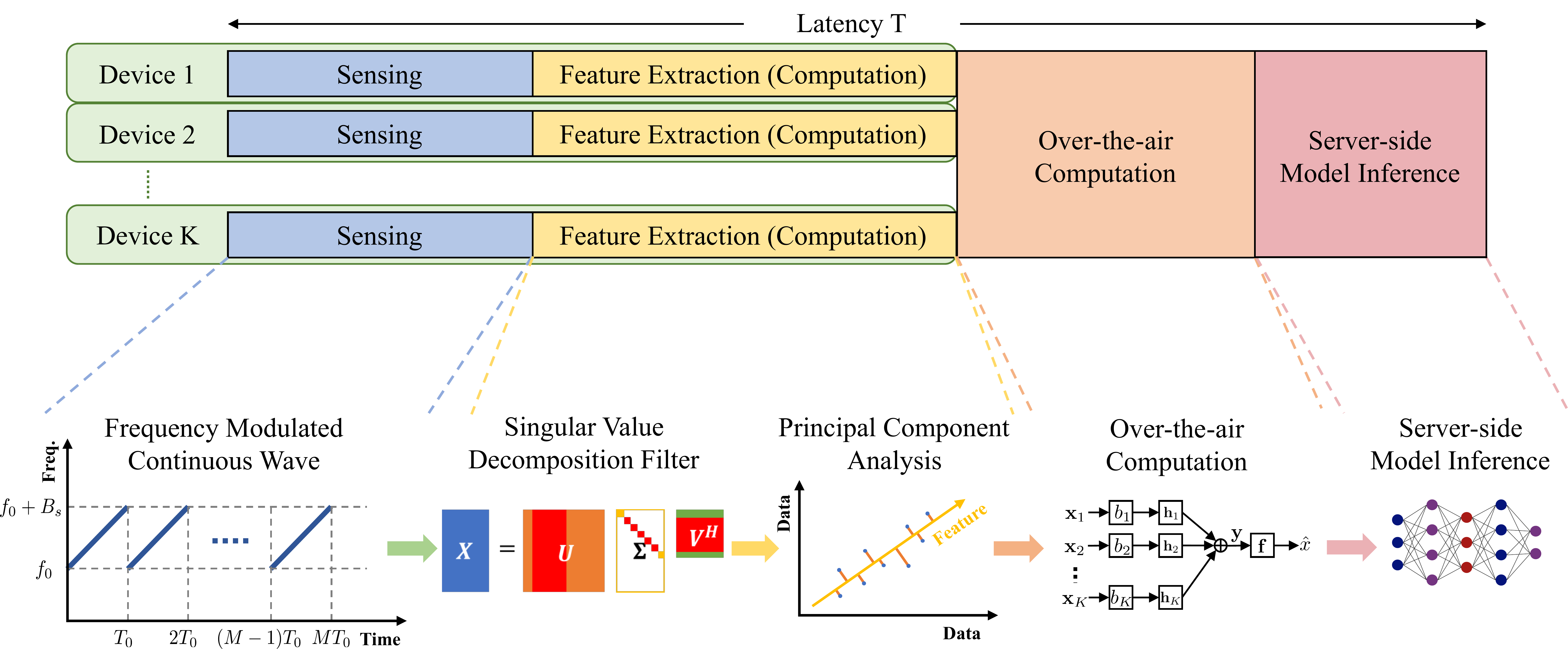}
    \caption{The workflow for completing an inference task.}
    \label{fig:time}
\end{figure*}

\subsection{Sensing Signal Processing and Feature Extraction}
We adopt the models of sensing signal processing and feature extraction proposed in \cite{wen2022task1}. As shown in Fig.~\ref{fig:time}, during the radar sensing stage, each device 
transmits the FMCW signal of $N$ up-ramp chirps for sensing. Each chirp has a time duration of $T_0=T_s/N$ with $T_s$ being the total sensing time. For device $k$, the sensing signal of one chirp is formulated as 
\begin{equation}
    c_{s,k}(t) = \operatorname{rect}\left(\frac{t}{T_0}\right) \cdot \cos \left(2 \pi f_{k,0} t+2\pi \frac{B_s}{T_0} t^{2}\right),\; 1\leq k \leq K,
\end{equation}
where $\operatorname{rect}\left(\cdot\right)$ is the rectangular pulse function with amplitude 1 and pulse length 1 centered at $t=0$, $f_{k,0}$ is the starting frequency of sensing signal, $B_s$ is the bandwidth of the sensing signal. It follows that the signal of the whole sensing duration is  
\begin{equation}
    s_k(t) = \sum_{n=0}^{N-1} c_{s,k}(t-n T_0),
\end{equation}
\begin{equation}
\begin{aligned}
    s_k(t) = \sum_{n=0}^{N-1} &\operatorname{rect}\left(\frac{t-n T_0}{T_0}\right) \\ 
    \times &\cos \left(2 \pi f_0(t-n T_0)+2\pi \frac{B_s}{T_0} (t-n T_0)^{2}\right).
\end{aligned}
\end{equation}
Then the reflected signals from the direct and indirect paths are received by each device. The desirable echo signal is the one directly reflected from the target, given by
\begin{equation}
    u_k(t)=H_{s,k}(t)s_k(t-\tau),\; 1\leq k \leq K,
\end{equation}
where $H_{s,k}(t)$ is the reflection matrix of the target including the round-trip path-loss, $\tau$ is the round-trip delay. The echo signal indirectly reflected through the $j$-th indirect path is 
\begin{equation}
    v_{k,j}(t) = C_{r,k,j}(t)s_k(t-\tau_j),\; 1\leq k \leq K,
\end{equation}
where $C_{s,k,j}(t)$ is the round-trip coefficient of path $j$, $\tau_j$ is the delay of the $j$-th path. Note that $H_{s,k}(t)$ and $\{ C_{s,k,j}(t) \}$ can be pre-estimated by each device and fed back to the edge server before the inference task. Thereby, the received signal of ISAC device $k$ is given by
\begin{equation}\label{eq:r}
    r_{k}(t)=u_{k}(t)+\sum_{j=1}^{J} v_{k, j}(t)+n_{r}(t),\; 1\leq k \leq K,
\end{equation}
where $u_{k}(t)$ is the desired signal for completing the inference task, $\sum_{j=1}^{J} v_{k, j}(t)$ is the clutter of $J$ indirect reflection paths and $n_r(t)$ is the white Gaussian noise. In \eqref{eq:r}, the useful signal $u_{k}(t)$ is polluted by the additive sensing clutter and noise. In the sequel, the clutter cancellation procedure is introduced. 

\subsubsection{Clutter cancellation}
First, the received signal of device $k$ is sampled at a frequency of $f_s$ into a complex feature vector $\rb_{k}\in\Cbb^{NT_0f_s}$. The data sample vector $\mathbf{r}_k$ contains both the ranging and velocity information of the target. Thus, for deriving the information of sensing target, $\mathbf{r}_k$ is transformed into a complex matrix $\mathbf{R}_{k}\in\mathbb{C}^{T_0f_s\times N}$, the column dimension of which is usually used for ranging and the row dimension contains the feature in the Doppler spectrum shift. Each column of $\mathbf{R}_k$ represents the data samples in one chirp containing the distance information of the target and each row of $\mathbf{R}_k$ reflects the motion of the target among different chirps, where the velocity of the target can be extracted from the Doppler shift.
Then, the SVD based linear filter proposed in \cite{cortes1995support} is utilized for clutter cancellation. To be specific, the SVD of $\Rb_{k}$ is 
\begin{equation}
    \Rb_{k} = \Ub\Sigmab\Vb^H = \sum_{i=1}^{I}\ub_i\sigma_i\vb_i^H,\; 1\leq k \leq K,
\end{equation}
where $I=\min\{T_0f_s, N\}$, $\ub_i$, $\sigma_i$ and $\vb_i$ are the $i$-th left singular vector, singular value and right singular vector of $\Rb_{k}$, respectively, $\Vb^H$ is the conjugate transpose of $\Vb$. Clutter cancellation is performed by deleting the principal and least dimensions of $\Rb_{k}$. As a result, the data matrix after filtering is
\begin{equation}
    \tilde{\Rb}_{k} = \sum_{i=r_1}^{r_2}\ub_i\sigma_i\vb_i^H,\; 1\leq k \leq K,
\end{equation}
where $1\leq r_1 $ and $ r_2\leq I$ are empirical parameters with respect to different kinds of radar sensors. Since only the information in row dimension, i.e., the Doppler spectrum shift, is needed for the inference task, $\tilde{\Rb}_{k}$ is compressed into a vector $\bar{\rb}_{k}\in\Cbb^{N}$. Its $i$-th element is given by
\begin{equation}
    \bar{r}_{k}^i = \sum_{j=1}^{T_0f_s}\tilde{R}_{k}^{j,i},\; 1\leq k \leq K,
\end{equation}
where $\tilde{R}_{k}^{j,i}$ is the $(j,i)$-th element of matrix $\tilde{\Rb}_{k}$. Then the real part and the imaginary part of $\bar{r}_k$ is cascaded into a real vector $\tilde{\rb}_{k}\in\Rbb^{2N}$
\begin{equation}
    \tilde{\rb}_{k} = [\mathfrak{Re}(\bar{\rb}_{k}), \mathfrak{Im}(\bar{\rb}_{k})]
\end{equation}

\subsubsection{Feature extraction}
 Following \cite{lan2022progressive, wen2022task1, wen2022task2}, the PCA based linear extractor is used to extract the local feature vector from clutter-cancelled sensory data $\tilde{\rb}_{k}\in\Rbb^{2N}$. The PCA is performed at the edge server before the inference task using the training dataset. Then, the template of the $M$ principal eigen-subspace is broadcast to all devices for extracting the local feature vectors $\{\tilde{\rb}_k \in \Rbb^M\}$ with $M$ being the number of extracted feature elements. Since the clutter cancellation and feature extraction processes are linear and based on \eqref{eq:r}, the $m$-th feature element of $\tilde{\rb}_k$ is given by 
\begin{equation}
    \tilde{r}_k(m) = \tilde{u}_k(m)+\sum_{j=1}^J \tilde{v}_{k,j}(m)+n_r(m), 
\end{equation}
where $\tilde{u}_k(m)$ is the ground-truth of feature $m$, $\tilde{v}_{k,j}(m)$ is the clutter from path $j$, $n_r(m)$ is the noise in Gaussian distribution, defined by
\begin{equation}\label{dist:noise}
    n_r(m)\sim\Ncal\left(0,\sigma_r^2\right),\; 1\leq m \leq M.
\end{equation}
Next, each feature element of device $k$ is normalized by its sensing power $P_{s, k}$ and the normalized feature element $m$ is given by
\begin{equation}\label{def:x}
    x_k(m)=\frac{\tilde{r}_{k}(m)}{\sqrt{P_{s, k}}}=x(m)+\tilde{c}_{s, k}(m)+\frac{n_r(m)}{\sqrt{P_{s, k}}},
\end{equation}
where $x(m)=\tilde{u}_k(m) / \sqrt{P_{s, k}}$ is the normalized ground-truth feature and 
\begin{equation}
    \tilde{c}_{s, k}(m)=\sum_{j=1}^{J} \frac{\tilde{v}_{k, j}(m)}{\sqrt{P_{s, k}}}, \; 1\leq m \leq M, \; 1\leq k \leq K,
\end{equation}
is the normalized clutter. Since clutter is rich scattering and its number of paths $J$ is very large, these individual clutter elements are assumed to be independent and identically distributed with finite variance. Thus $\tilde{c}_{s,k}(m)$ follows a Gaussian distribution according to the Central Limit Theorem (CLT), given by 
\begin{equation}\label{dist:clutter}
    \tilde{c}_{s, k}(m)\sim\mathcal{N}\left(\mu_{s,k},\sigma_{s, k}^2\right), \; 1\leq m \leq M, \; 1\leq k \leq K, 
\end{equation}
where $\mu_{s,k}$ is the mean of clutter and can be pre-estimated and $\sigma_{s, k}^2$ is the clutter variance. Then the pre-estimated mean of $\tilde{c}_{s,k}(m)$ is eliminated to derive a zero-mean residual clutter element $c_{s, k}(m) = \tilde{c}_{s,k}(m) - \mu_{s,k}$. The CLT states that the sum or mean of a large number of independent and identically distributed random variables will approximate a Gaussian distribution, regardless of the shape of the original distribution, as long as the original variables have finite variance. Thereby, the local feature vector of device $k$ can be written as
\begin{equation}
    \xb_k = \xb + {\bf c}_{s,k} + \dfrac{\nb_r}{ \sqrt{P_{s,k}}}, \; 1\leq k \leq K, 
\end{equation}
where $\xb = \{x(m)\}_{m=1}^M$, ${\bf c}_{s,k}(m) = \{c_{s, k}\}_{m=1}^M$ and $\nb_r = \{ n_r(m)\}_{m=1}^M$.

\subsection{Feature Distribution}\label{Sect:FeatureDistribution}
Consider a classification task with $L$ classes. Following \cite{lan2022progressive,wen2022task1,wen2022task2}, the ground-truth feature vector $\xb$ is assumed to follow a Gaussian mixture distribution. Since PCA is performed, different elements of ground-truth feature vector are independent. Consider an arbitrary element $x(m)$, its distribution is given as
\begin{equation}\label{def:GTx}
    f(x(m)) = \frac{1}{L}\sum_{\ell=1}^L f_\ell(x(m)), \; 1\leq m \leq M,
\end{equation}
where $f_\ell(x(m)) = \Ncal\left(\mu_{\ell,m},\sigma_m^2\right)$ is the probability density function of the Gaussian component corresponding to the $\ell$-th class, $\mu_{\ell,m}$ is the centroid of class $\ell$ and $\sigma_m^2$ is the variance. These parameters are pre-estimated using the training dataset. Based on \eqref{def:GTx} and the clutter distribution in \eqref{dist:clutter} and the noise distribution in \eqref{dist:noise}, the distribution of the local feature element $x_k(m)$ can be derived as in the following lemma.
\begin{lemma}\label{lemma:1}
The distribution of local feature elements $x_k(m)$ can be derived as
\begin{equation}\label{dist:xtrans}
    x_k(m) \sim \frac{1}{L}\sum_{\ell=1}^L\Ncal\left(\mu_{\ell,m}, \sigma_m^2+\sigma_{s,k}^2+\frac{\sigma_r^2}{P_{s, k}}\right),~ 1\leq k\leq K.
\end{equation}
\end{lemma}
\begin{proof} 
    See Appendix~\ref{appendix:1}.
\end{proof}

\subsection{Broadband Over-the-air Computation}
In the edge-device co-inference system shown in Fig.~\ref{fig:system}. The edge server needs to aggregate all local feature vectors to obtain a global denoised one. If the conventional orthogonal multiple access technique such as TDMA is used, the consumed resource blocks linearly increase with the number of devices, leading to heavy communication overhead. To address this communication bottleneck, the technique of AirComp (see \cite{wang2022over,zhu2018mimo, wen2019reduced, zhu2020one}) is adopted for the feature vector aggregation. As shown in Fig.~\ref{fig:signal}, over the same subcarrier, it allows all devices simultaneously transmit the same feature dimension. At the server, the waveform superposition property is leveraged to directly derive a weighted sum of the elements from all devices. As a result, the communication overhead remains unchanged as the number of devices varies, leading to a significant enhancement of communication efficiency. 
\begin{figure}[t]
    \centering
    \includegraphics[width=\columnwidth]{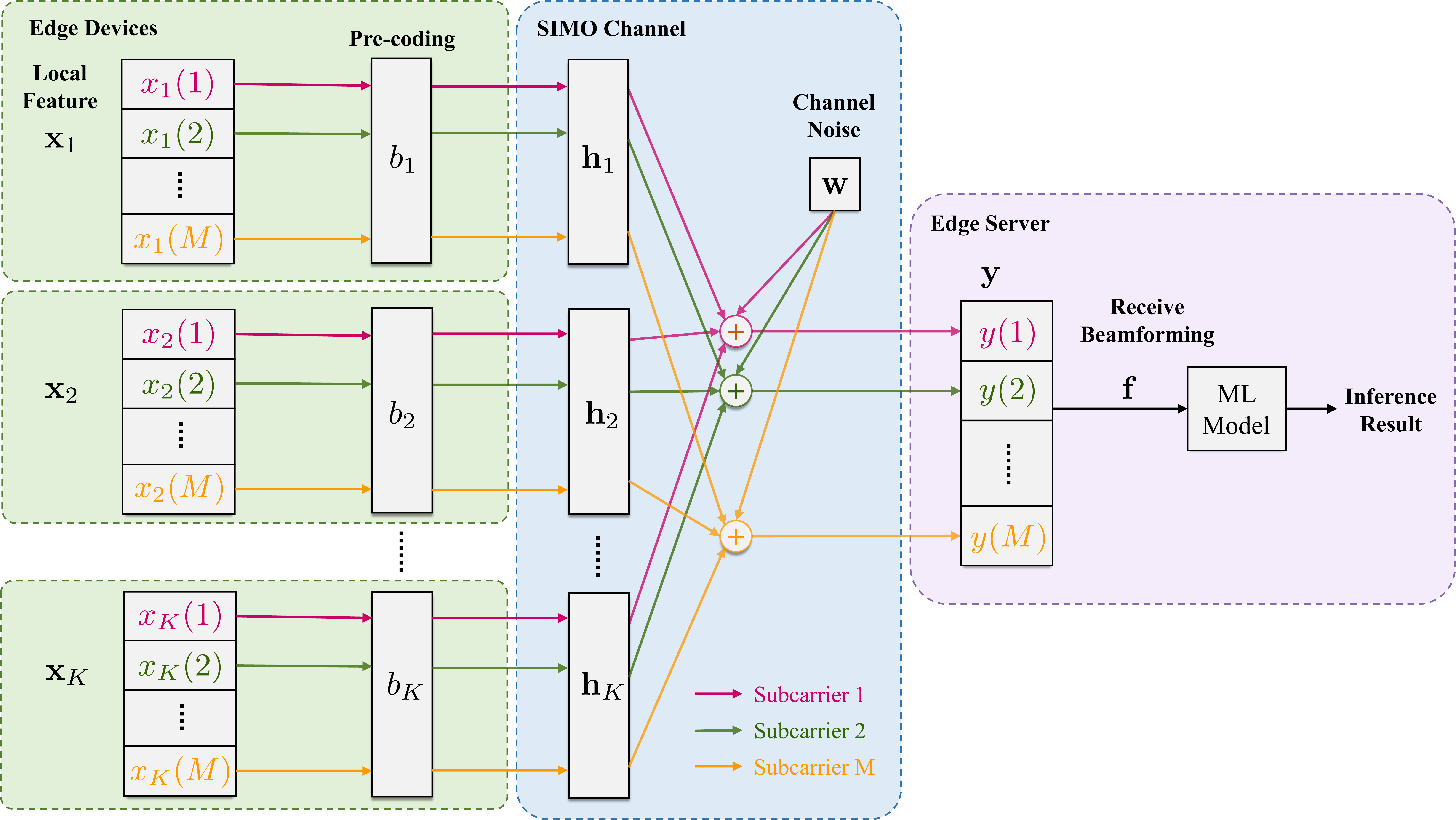}
    \caption{The signal diagram of over-the-air computation with OFDM.}
    \label{fig:signal}
\end{figure}

Specifically, consider an arbitrary subcarrier to aggregate an arbitrary feature dimension $m$. At each device, the local feature element $x_k(m)$ is first pre-coded with $b_{k,m}$ and then transmitted over the single-input-multiple-output (SIMO) channel, the aggregated received signal at the server is given by 
\begin{equation}\label{eq:channel}
\yb(m)=\sum_{k=1}^K\hb_{k,m}b_{k,m}x_k(m)+\wb(m),
\end{equation}
where $\hb_{k,m}\in \Cbb^{N_r}$ is the channel gain of device $k$, $b_{k,m}$ is the pre-coding complex scalar of $x_k(m)$, $\wb(m)$ is the additive white Gaussian noise following the distribution of $\Ncal(\mathbf{0},N_0\Ib)$ and $N_0$ is the channel noise variance, $\Ib\in\Rbb^{N_r\times N_r}$ is the identity matrix. As mentioned, the channel vector $\hb_{k,m}$ remains static for aggregating all feature elements. After receiving the signal, a receive beamforming vector $\fb_m\in\Cbb^{N_r}$ is added by the edge server to extract the feature vector
\begin{equation}\label{Eq:ReceivedElement}
    \hat{x}(m)=\fb_m^H\yb(m)=\fb_m^H\sum_{k=1}^K\hb_{k,m}b_{k,m}x_k(m)+\fb_m^H\wb(m).
\end{equation}
For similar reasons as \eqref{dist:xtrans}, the distribution of $\hat{x}(m)$ can be further derived as 
\begin{equation}\label{Eq:ReceivedElementDistribution}
    f\left(\hat{x}(m) \right) = \dfrac{1}{L} \sum\limits_{\ell=1}^L f_{\ell}(\hat{x}(m)), \; 1\leq m \leq M,
\end{equation}
where 
\begin{equation}
    f_{\ell}(\hat{x}(m)) = \mathbf{f}_m^H\sum_{k=1}^K\mathbf{h}_{k,m}b_{k,m} f_{\ell}(x_k(m)) + \mathbf{f}_m^H f(\mathbf{w}(m)),
\end{equation}
and $f_{\ell}(x_k(m)) = \mathcal{N}\left(\mu_{\ell,m}, \sigma_m^2 + \sigma_{s,k}^2 + \frac{\sigma_r^2}{P_{s, k}}\right)$ and $f(\mathbf{w}(m)) = \mathcal{N}(\mathbf{0},N_0\mathbf{I})$ are the distributions of the $\ell$-th component of local feature in device $k$ and the Gaussian white noise in wireless channel.

Then, all dimensions of the local feature vectors are aggregated in a similar way over $M$ subcarriers, as shown in Fig. \ref{fig:signal}. Thereby, the overall received feature vector is $\hat{\bf x}= [ \hat{x}(1),..., \hat{x}(m), ..., \hat{x}(M)]^T$. Since PCA is performed at each device, different elements of each local feature vector are independent. As a result, the distributions of different elements in the received feature vector $\hat{\xb}$ are independent, since each feature element $\hat{x}(m)$ only depends on the corresponding local feature elements $\{x_k(m)\}$ and the white Gaussian channel noise according to \eqref{Eq:ReceivedElement}.



\section{Problem Formulation and Simplification}\label{sec:3}
In this section, a novel design criterion called minimum pair-wise discriminant gain is adopted, based on which, the problem is formulated.

\subsection{Minimum Pair-Wise Discriminant Gain}

As mentioned, the design criterion adopted in this work is maximum inference accuracy instead of the conventional minimum mean square error (MMSE), as the latter cannot distinguish the importance levels of different elements to the inference task \cite{wen2022task2}. However, the instantaneous inference accuracy is unknown and does not have a mathematical model at the design stage. To this end, an approximate but tractable metric called discriminant gain is adopted as an alternative. Based on the received feature distribution in \eqref{Eq:ReceivedElementDistribution}, a pair-wise discriminant gain of an arbitrary class pair $(\ell, \ell^{'})$ is defined as the symmetric KL divergence of their corresponding Gaussian components \cite{kullback1951information, lan2022progressive}. Specifically, considering the $m$-th feature element, its pair-wise discriminant gain in terms of the class pair $(\ell, \ell^{'})$ is given by 
\begin{equation}
\begin{aligned}
    G_{\ell, \ell^{\prime}}(\hat{x}(m))\triangleq &D_{KL}\left[f_\ell\left(\hat{x}(m)\right) \| f_{\ell^\prime}\left(\hat{x}(m)\right)\right]\\
    &+ D_{KL}\left[f_{\ell^\prime}\left(\hat{x}(m)\right) \| f_\ell\left(\hat{x}(m)\right)\right],\\
    =&\int_{\hat{x}(m)} \left[f_{\ell}\left(\hat{x}(m)\right) \log \left[\frac{f_{\ell}\left(\hat{x}(m)\right)}{f_{\ell^\prime}\left(\hat{x}(m)\right)}\right]\right.\\
    &\left.+ f_{\ell^{\prime}}\left(\hat{x}(m)\right) \log \left[\frac{f_{\ell^\prime}\left(\hat{x}(m)\right)}{f_{\ell}\left(\hat{x}(m)\right)}\right]\right] \mathrm{d} \hat{x}(m),\\
\end{aligned}
\end{equation}
where $D_{KL}\left[\mathrm{p}\|\mathrm{q}\right]$ represents the KL divergence between distributions $\mathrm{p}$ and $\mathrm{q}$. As mentioned, different feature elements in the received feature vector $\hat{\xb}$ are independent. It follows that the pair-wise discriminant gain of $\hat{\xb}$ is derived as
\begin{equation}
\begin{aligned}
    G_{\ell, \ell^{\prime}}(\hat{\mathbf{x}})&=D_{K L}\left[f_{\ell}(\hat{\mathbf{x}}) \| f_{\ell^{\prime}}(\hat{\mathbf{x}})\right] + D_{K L}\left[f_{\ell^{\prime}}(\hat{\mathbf{x}}) \| f_{\ell}(\hat{\mathbf{x}})\right]\\
    &=\sum_{m=1}^{M} G_{\ell, \ell^{\prime}}\left(\hat{x}(m)\right),\;\forall (\ell,\ell^{'}).
\end{aligned}
\end{equation}
With a larger pair-wise discriminant gain, the corresponding pair of classes are better separated in the feature space, thus resulting in an improved achievable inference accuracy. 

In existing literatures \cite{lan2022progressive, wen2022task1, wen2022task2}, maximizing the average of all pair-wise discriminant gains as defined in \eqref{Eq:AverageDG} is used as the design criterion, i.e., 
\begin{equation}\label{Eq:AverageDG}
    G\left(\hat{\xb}\right)=\frac{2}{L(L-1)} \sum_{\ell^{\prime}=1}^{L} \sum_{\ell<\ell^{\prime}}G_{\ell, \ell^{\prime}}(\hat{\xb}).
\end{equation}
However, under this design goal, the values of one or several pair-wise discriminant gains can be dominant, while other pair-wise discriminant gains are very small. That says, only a subset of class pairs is well separated but the others cannot be differentiated [see Fig.~\ref{figgr:DG}(a) for example]. This leads to an unbalanced and low inference accuracy. To overcome this limitation, this work proposes to maximize the minimum pair-wise discriminant gain of all pairs, defined as
\begin{equation}\label{Eq:MimimumDG}
\begin{aligned}
    G_{\min}\left(\hat{\xb}\right)&=\min_{1 \leq \ell \neq \ell^{\prime} \leq L} G_{\ell, \ell^{\prime}}(\hat{\xb})\\
    &=\min_{1 \leq \ell \neq \ell^{\prime} \leq L} \sum_{m=1}^{M} G_{\ell, \ell^{\prime}}\left(\hat{x}(m)\right),\;\forall (\ell,\ell^{'}).
\end{aligned}
\end{equation}
By maximizing the minimum pair-wise discriminant gain in \eqref{Eq:MimimumDG}, the closest class pair in the feature space can be well separated, leading to a balanced and enhanced inference accuracy.


\subsection{Problem Formulation}
The maximization of the minimum pair-wise discriminant gain defined in \eqref{Eq:MimimumDG} is constrained by the energy threshold of each device. Consider an arbitrary device $k$, its
sensing energy consumption is $P_{s, k}T_{s, k}$ with $P_{s,k}$ being the sensing power and $T_{s, k}$ being the fixed sensing time. Its energy consumption for on-device feature extraction is denoted as $E_{p,k}$, which is a constant. For AirComp, the power of device $k$ to transmit the $m$-th feature element is 
\begin{equation}\label{def:comm_power}
     P_{c,k}(m) = b_{k,m}\Ebb\left[x_k(m) x_k(m)^H\right]b_{k,m}^H,\; \forall (m, k).
\end{equation}
In \eqref{def:comm_power}, since the distribution of $x_k(m)$ is known [Please refer to \eqref{def:GTx}], its variance is determined and is denoted as $X_{k}(m) = \Ebb\left[x_k(m) x_k(m)^H\right]$. It follows that the energy consumption of the whole AirComp process is 
\begin{equation}
    E_{c,k} = T_c P_{c,k}(m) = T_c\sum_{m=1}^M b_{k,m}b_{k,m}^H X_{k}(m),\; 1\leq k \leq K,
\end{equation}
where $T_c$ is the AirComp transmission time for each element. Therefore, the energy consumption constraint of device $k$ can be derived as 
\begin{equation}\label{cons:energy1}
    P_{s, k}T_{s, k}+E_{p, k}+ T_c \sum_{m=1}^M b_{k,m}b_{k,m}^H X_{k}(m) \leq E_{k},\; 1 \leq k \leq K,
\end{equation}
where $E_{k}$ is the energy threshold of device $k$.

Accordingly, the problem of maximizing the minimum pair-wise discriminant gain under the energy consumption constraint can be formulated as
\begin{equation}\Prob{1}
\begin{aligned}
    &\max_{\substack{\{P_{s, k}\},\\\{b_{k,m}\},\{\fb_m\}}}&&\!\!\!\!\min_{1 \leq \ell \neq \ell^{\prime} \leq L}\;\; \sum_{m=1}^{M} G_{\ell, \ell^{\prime}}\left(\hat{x}(m)\right),\\
    &\qquad\text{s.t.}&&\!\!\!\!P_{s, k}T_{s, k}+E_{p, k}+ T_c \sum_{m=1}^M b_{k,m}b_{k,m}^H X_{k}(m)\\
    &&&\qquad\qquad\qquad\qquad\quad\leq E_{k},\; 1 \leq k \leq K.
\end{aligned}
\end{equation}

\subsection{Problem Simplification}
Since the distributions of the received elements $\{\hat{x}(m)\}$ in \eqref{Eq:ReceivedElementDistribution} are complex, the minimum pair-wise discriminant gain defined based on these distributions, i.e., the objective of $\Prob{1}$ is a complicated non-convex function. Besides, the energy constraint in $\Prob{1}$ is also non-convex. To address this complicated non-convex problem, a conventional approach (see, \cite{zhu2018mimo, wen2019reduced, wen2022task2}) is applied to simplify it by pre-determining the precoders as
\begin{equation}\label{def:c}
    \fb_m^H\hb_{k,m}b_{k,m} = c_{k,m},\; 1\leq m \leq M,\; 1 \leq k \leq K,
\end{equation}
where $c_{k,m} \in \Rbb^{+}$ represents the received signal power of element $m$ from device $k$. Accordingly, the precoder $b_{k,m}$ can be written in a function of $c_{k,m}$ by multiplying $(\mathbf{f}_m^H\mathbf{h}_{k,m})^H$ on both sides of equation \eqref{def:c}:
\begin{equation}
    (\mathbf{f}_m^H\mathbf{h}_{k,m})^H \mathbf{f}_m^H \mathbf{h}_{k,m} b_{k,m} = \mathbf{h}_{k,m}^H \mathbf{f}_m c_{k,m},\; \forall (m, k).
\end{equation}
Then, $b_{k,m}$ is derived as
\begin{equation}\label{Eq:Precoder}
    b_{k,m} = \frac{c_{k,m}\hb_{k,m}^H \fb_m}{\hb_{k,m}^H \fb_m \fb_m^H \hb_{k,m}} = \frac{c_{k,m}}{\fb_m^H \hb_{k,m}},\; \forall (m, k).
\end{equation}
By substituting $b_{k,m}$ in \eqref{Eq:Precoder} into the received feature element in \eqref{Eq:ReceivedElement}, we have
\begin{equation}
    \hat{x}(m)= \sum_{k=1}^K c_{k,m}x_k(m)+\fb_m^H\wb(m), \; \forall (m, k),
\end{equation}
which, by substituting the local feature elements $\{x_k(m)\}$ in \eqref{def:x}, is further derived as
\begin{equation}
\begin{aligned}
    \hat{x}(m) = &\left(\sum_{k=1}^K c_{k,m}\right) x(m) \\
    &+ \sum_{k=1}^Kc_{k,m}\left(c_{s, k}(m)+\frac{n_r(m)}{\sqrt{P_{s, k}}}\right) + \fb_m^H\wb(m).
\end{aligned}
\end{equation}
It follows that the distribution of $\hat{x}(m)$ can be derived as
\begin{equation}
\begin{aligned}
    f\left(\hat{x}(m) \right) &= \dfrac{1}{L} \sum\limits_{\ell=1}^L f_{\ell}(\hat{x}(m)) \\
    &= \dfrac{1}{L} \sum\limits_{\ell=1}^L \mathcal{N}\left(\hat{\mu}_{\ell,m}, \hat{\sigma}^2_m\right), \; 1\leq m \leq M,
\end{aligned}
\end{equation}
Since the transformations in (33) are all linear and $x_\ell(m)\sim\mathcal{N}(\mu_{\ell,m},\sigma_m^2)$, $c_{s,k}(m)\sim\mathcal{N}(0,\sigma_{s,k}^2)$ and $n_r(m)\sim\mathcal{N}(0,\sigma_r^2)$ are following independent Gaussian distributions, the distribution of $\hat{x}(m)$ can be derived in a closed form. The mean of the $\ell$-th class component is given as follows:
\begin{equation}
    \hat{\mu}_{\ell,m} = \left(\sum_{k=1}^K c_{k,m}\right) \mu_{\ell,m},
\end{equation}
and the variance of the $\ell$-th class component is given as follows:
\begin{equation}
\begin{aligned}
    \hat{\sigma}^2_m =& \left(\sum_{k=1}^K c_{k,m}\right)^2 \sigma_m^2 \\
    &+\sum_{k=1}^K c_{k,m}^2\left(\sigma_{s, k}^2+\frac{\sigma_r^2}{P_{s, k}}\right)+N_0\mathbf{f}_m^H\mathbf{f}_m.
\end{aligned}
\end{equation}
As a result, the pair-wise discriminant gain $ G_{\ell, \ell^{\prime}}\left(\hat{x}(m)\right)$ can be derived as
\begin{equation}
\begin{aligned}
    &G_{\ell, \ell^{\prime}}\left(\hat{x}(m)\right) = \frac{\left(\hat{\mu}_{\ell,m}-\hat{\mu}_{\ell^\prime,m}\right)^2}{\hat{\sigma}_m^2}= \\
    &\frac{\left(\mu_{\ell,m}-\mu_{\ell^\prime,m}\right)^2\left(\sum_{k=1}^K c_{k,m} \right)^2}{\sigma_m^2\left(\sum_{k=1}^K c_{k,m}\right)^2+\sum_{k=1}^K c_{k,m}^2\left(\sigma_{s, k}^2+\frac{\sigma_r^2}{P_{s, k}}\right)+N_0\fb_m^H\fb_m}.
\end{aligned}
\end{equation}
Besides, by substituting the precoders in \eqref{Eq:Precoder} into the energy constraint in $\Prob{1}$, it can be re-formulated as
\begin{equation}\label{cons:energy2}
    P_{s, k}T_{s, k}+E_{p, k}+ T_c \sum_{m=1}^M \frac{c_{k,m}^2 X_{k}(m)}{\hb_{k,m}^H \fb_m \fb_m^H \hb_{k,m}} \leq E_k,~ 1 \leq k \leq K.
\end{equation}

In summary, with the precoders defined in \eqref{Eq:Precoder}, $\Prob{1}$ is simplified as
\begin{equation}\Prob{2}
\begin{aligned}
    &\max_{\substack{\{P_{s, k}\},\\\{c_{k,m}\},\{\fb_m\}}}&&\!\!\!\!\min_{1 \leq \ell \neq \ell^{\prime} \leq L}\frac{\left(\hat{\mu}_{\ell,m}-\hat{\mu}_{\ell^\prime,m}\right)^2}{\hat{\sigma}_m^2},\\
    &\qquad\text{s.t.}&&\!\!\!\!P_{s, k}T_{s, k}+E_{p, k} + T_c \\
    &&&\,\times \sum_{m=1}^M\frac{c_{k,m}^2X_{k}(m)}{\hb_{k,m}^H \fb_m \fb_m^H \hb_{k,m}} \leq E_k ,~ 1\leq k\leq K.
\end{aligned}
\end{equation}
\,

\section{Joint Sensing Power Assignment, Transmit Precoding and Receive Beamforming}\label{sec:4}

Although $\Prob{2}$ has a simplified form, it is still difficult to solve due to the minimax form and the complicated non-convex fractional functions in both objective and constraints. To address this problem, in the sequel, variables transformation is conducted to decouple the minimax objective function and to derive an equivalent problem with the d.c. form, based on which, the typical method of SCA is utilized to obtain a sub-optimal solution. 

\subsection{Variables Transformation}
To begin with, the following variable is defined to decouple the minimax objective function:
\begin{equation}
    \alpha =\min_{1 \leq \ell \neq \ell^{\prime} \leq L}\;\; \sum_{m=1}^{M} G_{\ell, \ell^{\prime}}\left(\hat{x}(m)\right). 
\end{equation}
It follows that all pair-wise discriminant gains should be no less than $\alpha$:
\begin{equation}\label{Eq:DGConstraint}
    \sum_{m=1}^M \frac{\left(\hat{\mu}_{\ell,m}-\hat{\mu}_{\ell^\prime,m}\right)^2}{\hat{\sigma}_m^2} \geq \alpha,\; 1 \leq \ell \neq \ell^\prime \leq L.
\end{equation}
Accordingly, $\Prob{2}$ is equivalent to the problem that maximizes $\alpha$ under the constraints of the original energy consumption and pair-wise discriminant gains in \eqref{Eq:DGConstraint}, i.e., 
\begin{equation*}\Prob{3}
    \begin{aligned}
        &\max_{\substack{\{P_{s, k}\}, \alpha,\\ \{c_{k,m}\},\{\fb_m\}}}&& \alpha,\\
        &\quad\ \ \text{s.t.}&& P_{s, k}T_{s, k}+E_{p, k} + T_c\\
        &&&\,\times \sum_{m=1}^M\frac{c_{k,m}^2X_{k}(m)}{\hb_{k,m}^H \fb_m \fb_m^H \hb_{k,m}} \leq E_k ,~ 1\leq k\leq K,\label{cons:o1}\\
        &&& \frac{\left(\hat{\mu}_{\ell,m}-\hat{\mu}_{\ell^\prime,m}\right)^2}{\hat{\sigma}_m^2} \geq \alpha ,\ \,\qquad 1 \leq \ell \neq \ell^\prime \leq L.
    \end{aligned}
\end{equation*}
Then, to further address the non-convex ratios in the energy consumption constraint (the first constraint), the following variables are introduced:
\begin{equation}\label{def:u}
    u_{k,m} = \frac{c_{k,m}^2}{\hb_{k,m}^H \fb_m \fb_m^H \hb_{k,m}}\geq0, \; 1 \leq m \leq M.
\end{equation}
By substituting \eqref{def:u}, the energy constraint in $\Prob{3}$ for each device $k$ is equivalently decomposed into the following two constraints:
\begin{equation}\label{cons:o1a}
P_{s, k}T_{s, k}+E_{p, k}+ T_c \sum_{m=1}^M u_{k,m} X_{k}(m) \leq E_k,~1 \leq k \leq K, 
\end{equation}
and 
\begin{equation}\label{cons:o1b}
c_{k,m}^2 = \hb_{k,m}^H \fb_m \fb_m^H \hb_{k,m}u_{k,m},~1 \leq m \leq M.
\end{equation}
Next, we extend the feasible region of the equality constraint \eqref{cons:o1b} as in \eqref{cons:o1be} while keeping the same optimal solution to $\Prob{3}$, as shown in Lemma~\ref{lemma:2}.
\begin{equation}\label{cons:o1be}
    c_{k,m}^2 \leq \hb_{k,m}^H \fb_m \fb_m^H \hb_{k,m}u_{k,m},~1 \leq m \leq M.
\end{equation}
\begin{lemma}\label{lemma:2}
    A new problem $\Prob{3^\prime}$ which extends the feasible region of \eqref{cons:o1b} to \eqref{cons:o1be} and keeps the same objective function, the constraint in \eqref{cons:o1a} and the pair-wise discriminant constraint (the second constraint in $\Prob{3}$), reaches the same optimum as $\Prob{3}$.
\end{lemma}
\begin{proof}
    See Appendix \ref{appendix:2}.
\end{proof}

To further address the non-convex pair-wise discriminant gain constraint (the second constraint) in $\Prob{3}$, a set of variables $\{v_{\ell,\ell^\prime,m}\}$ are introduced as follows:
\begin{equation}\label{cons:o2b}
\begin{aligned}
    &\frac{\left(\mu_{\ell,m}-\mu_{\ell^\prime,m}\right)^2}{v_{\ell,\ell^\prime,m}}\left(\sum_{k=1}^K c_{k,m}\right)^2 \\
    &\;\;= \sigma_m^2\left(\sum_{k=1}^K c_{k,m}\right)^2+\sum_{k=1}^K c_{k,m}^2\left(\sigma_{s, k}^2+\frac{\sigma_r^2}{P_{s, k}}\right)+N_0\fb_m^H\fb_m. 
\end{aligned}
\end{equation}
It follows that the pair-wise discriminant gain constraint (the second constraint) in $\Prob{3}$ can be equivalently decomposed as
\begin{equation}\label{cons:o2a}
    \sum_{m=1}^{M} v_{\ell,\ell^\prime,m} \geq \alpha.
\end{equation}
For similar reasons to \eqref{cons:o1b} and Lemma \ref{lemma:2}, the feasible region of the constraint in \eqref{cons:o2b} can be extended as that in \eqref{Eq:ExtR1} without changing the optimal solution of $\Prob{3}$.
\begin{equation}\label{Eq:ExtR1}
\begin{aligned}
    &\frac{\left(\mu_{\ell,m}-\mu_{\ell^\prime,m}\right)^2}{v_{\ell,\ell^\prime,m}}\left(\sum_{k=1}^K c_{k,m}\right)^2 \\
    &\;\;\geq \sigma_m^2\left(\sum_{k=1}^K c_{k,m}\right)^2+\sum_{k=1}^K c_{k,m}^2\left(\sigma_{s, k}^2+\frac{\sigma_r^2}{P_{s, k}}\right)+N_0\fb_m^H\fb_m. 
\end{aligned}
\end{equation}

In summary, $\Prob{3}$ can be equivalently derived as the following form:
\begin{equation*}\Prob{4}
    \begin{aligned}
        &\max_{\substack{\{P_{s, k}\},\{c_{k,m}\},\\ \{\fb_m\}, \{u_{k,m}\},\\ \{v_{\ell,\ell^\prime,m}\}, \alpha}}&& \alpha,\\
        &\qquad\,\text{s.t.}&& P_{s, k}T_{s, k}+E_{p, k}+ T_c \sum_{m=1}^M u_{k,m} X_{k}(m) \\
        &&& \qquad\qquad\qquad\quad - E_k \leq 0,~ 1\leq k\leq K,\\
        &&& \alpha - \sum_{m=1}^{M} v_{\ell,\ell^\prime,m} \leq 0 ,~\quad\; 1\leq \ell \neq \ell^\prime \leq L,\\
        &&& \frac{c_{k,m}^2}{u_{k,m}} - R_{k,m}\left(\fb_m\right) \leq 0,~ \forall (k, m),\\
        &&& Z_{m}\left(\{P_{s,k}\},\{c_{k,m}\},\fb_m\right) \\
        &&& \qquad\qquad - Q_{\ell,\ell^\prime,m}\left(\{c_{k,m}\},v_{\ell,\ell^\prime,m}\right) \leq 0,
    \end{aligned}
\end{equation*}
where 
\begin{equation*}
    \left\{
    \begin{aligned}
        & R_{k,m}\left(\fb_m\right)=\hb_{k,m}^H \fb_m \fb_m^H \hb_{k,m}, \\
        &Z_{m}\left(\{P_{s,k}\},\{c_{k,m}\},\fb_m\right) =\sigma_m^2\left(\sum_{k=1}^K c_{k,m}\right)^2 \\
        &\qquad\qquad\qquad\quad +\sum_{k=1}^K c_{k,m}^2\left(\sigma_{s, k}^2+\frac{\sigma_r^2}{P_{s, k}}\right)+N_0\fb_m^H\fb_m,\\
        &Q_{\ell,\ell^\prime,m}\left(\{c_{k,m}\},v_{\ell,\ell^\prime,m}\right) =\frac{\left(\mu_{\ell,m}-\mu_{\ell^\prime,m}\right)^2}{v_{\ell,\ell^\prime,m}}\left(\sum_{k=1}^K c_{k,m}\right)^2.
    \end{aligned}
    \right.
\end{equation*}
Although $\Prob{4}$ is still non-convex, it is in the d.c. form, as shown in Lemma \ref{lemma:3}.
\begin{lemma}\label{lemma:3}
    Problem $\Prob{4}$ is the d.c. problem.
\end{lemma}
\begin{proof}
    See Appendix~\ref{appendix:3}.
\end{proof}

\subsection{SCA based Algorithm}
To solve $\Prob{4}$, the method of SCA is adopted, which iterates between the following two steps until convergence to obtain a suboptimal solution, where all Karush-Kuhn-Tucker (KKT) conditions of $\Prob{4}$ are satisfied.
\begin{itemize}
    \item \emph{Convex approximation}: Based on a reference point, a convex approximation of $\Prob{4}$ is derived using Taylor expansion. The feasible region of the approximated problem is a subset of that of $\Prob{4}$. This guarantees that its solution is feasible for $\Prob{4}$. 

    \item \emph{Reference point update}: The approximated problem is optimally solved and the solution is used as the new reference point for the next iteration. 
\end{itemize}
In the sequel, the detailed procedures to solve $\Prob{4}$ are presented. 

\subsubsection{Convex approximation}
We first randomly initialize the optimization variables and set the counter $t=0$. Then, for an arbitrary iteration, i.e., $t>0$, the convex approximation of $\Prob{4}$ is described as follows. 

According to Lemma \ref{lemma:3}, $R_{k,m}\left(\fb_m\right)$ and $Q_{\ell,\ell^\prime,m}\left(\{c_{k,m}\},v_{\ell,\ell^\prime,m}\right)$ are both differentiable convex functions. Therefore, they are no less than their first-order Taylor expansions with the reference point being the optimal solution in the $(t-1)$-th iteration, i.e.,
\begin{align}
    &R_{k,m}(\fb_m) \geq \hat{R}_{k,m}^{[t]}(\fb_m), \\
    &Q_{\ell,\ell^\prime,m}\left(\{c_{k,m}\},v_{\ell,\ell^\prime,m}\right) \geq \hat{Q}_{\ell,\ell^\prime,m}^{[t]}\left(\{c_{k,m}\},v_{\ell,\ell^\prime,m}\right),
\end{align}
where $\hat{R}_{k,m}^{[t]}(\fb_m)$ and $\hat{Q}_{\ell,\ell^\prime,m}^{[t]}\left(\{c_{k,m}\},v_{\ell,\ell^\prime,m}\right)$ are the first-order Taylor expansions at $\fb_m^{[t]}$ and $\left(\{c_{k,m}^{[t]}\},v_{\ell,\ell^\prime,m}^{[t]}\right)$ respectively. They are given by
\begin{align}
    &\hat{R}_{k,m}^{[t]}(\fb_m) = R_{k,m}(\fb_m^{[t]}) + \left(\fb_m-\fb_m^{[t]}\right)^H \Ab_{k,m}^{[t]},\label{def:Rt}\\
    &\hat{Q}_{\ell,\ell^\prime,m}^{[t]}\left(\{c_{k,m}\},v_{\ell,\ell^\prime,m}\right) = Q_{\ell,\ell^\prime,m}\left(\left\{c_{k,m}^{[t]}\right\},v_{\ell,\ell^\prime,m}^{[t]}\right) \notag\\
    &\quad + B_{\ell,\ell^\prime,m}^{[t]}\left(v_{\ell,\ell^\prime,m}-v_{\ell,\ell^\prime,m}^{[t]}\right) +
    \sum_{k=1}^K C_{k,m}^{[t]}\left(c_{k,m}-c_{k,m}^{[t]}\right),\label{def:Qt}
\end{align}
where 
\begin{equation}
    \left\{\begin{aligned}
        &\Ab_{k,m}^{[t]} = \left.\frac{\partial R}{\partial \fb_m}\right|_{\fb_m=\fb_m^{[t]}} = 2 \hb_{k,m} \hb_{k,m}^H \fb_m^{[t]}, \\
        &B_{\ell,\ell^\prime,m}^{[t]} = \left.\frac{\partial Q}{\partial v_{\ell,\ell^\prime,m}}\right|_{v_{\ell,\ell^\prime,m}=v_{\ell,\ell^\prime,m}^{[t]}} \\
        &\qquad\quad= -\left(\frac{\left(\mu_{\ell,m}-\mu_{\ell^\prime,m}\right)\sum_{k=1}^{K} c_{k,m}^{[t]}}{v_{\ell,\ell^\prime,m}^{[t]}}\right)^2, \\
        &C_{k,m}^{[t]} = \left.\frac{\partial Q}{\partial c_{k,m}}\right|_{c_{k,m}=c_{k,m}^{[t]}} \\
        &\qquad\,= \frac{2 \left(\sum_{k=1}^{K} c_{k,m}^{[t]}\right)\left(\mu_{\ell,m}-\mu_{\ell^\prime,m}\right)^{2}}{v_{\ell,\ell^\prime,m}^{[t]}}.
    \end{aligned}\right.
\end{equation}
By replacing $R_{k,m}(\fb_m)$ and $Q_{\ell,\ell^\prime,m}\left(\{c_{k,m}\},v_{\ell,\ell^\prime,m}\right)$ with $\hat{R}_{k,m}^{[t]}(\fb_m)$ and $\hat{Q}_{\ell,\ell^\prime,m}^{[t]}\left(\{c_{k,m}\},v_{\ell,\ell^\prime,m}\right)$ respectively, an approximated convex problem of $\Prob{4}$ can be derived as
\begin{equation*}\Prob{5}
    \begin{aligned}
        &\max_{\substack{\{P_{s, k}\},\{c_{k,m}\},\\ \{\fb_m\},\{u_{k,m}\},\\ \{v_{\ell,\ell^\prime,m}\},\alpha}}&& \alpha,\\
        &\qquad\,\text{s.t.}&& P_{s, k}T_{s, k}+E_{p, k}+ T_c \sum_{m=1}^M u_{k,m} X_{k}(m) \\
        &&&\qquad\qquad\qquad\quad - E_k \leq 0,~ 1\leq k\leq K,\\
        &&& \alpha - \sum_{m=1}^{M} v_{\ell,\ell^\prime,m} \leq 0,~\quad\; 1\leq \ell \neq \ell^\prime \leq L,\\
        &&& \frac{c_{k,m}^2}{u_{k,m}} - \hat{R}_{k,m}^{[t]}\left(\fb_m\right) \leq 0,~ \forall (k,m),\\
        &&& Z_{m}\left(\{P_{s,k}\},\{c_{k,m}\},\fb_m\right) \\
        &&&\qquad\qquad - \hat{Q}_{\ell,\ell^\prime,m}^{[t]}\left(\{c_{k,m}\},v_{\ell,\ell^\prime,m}\right) \leq 0,
    \end{aligned}
\end{equation*}
where $Z_{m}\left(\{P_{s,k}\},\{c_{k,m}\},\fb_m\right)$, $\hat{R}_{k,m}^{[t]}\left(\fb_m\right)$ and $\hat{Q}_{\ell,\ell^\prime,m}^{[t]}\left(\{c_{k,m}\},v_{\ell,\ell^\prime,m}\right)$ are the same as those defined in $\Prob{4}$.

\subsubsection{Solution to $\Prob{5}$}
The primal-dual method is used to optimally solve $\Prob{5}$. First, the Lagrangian function of $\Prob{5}$ is given by
\begin{equation}
\begin{aligned}
    \Lcal_\Prob{5}&=-\alpha \\
    +&\sum_{k=1}^K \beta_k\left(P_{s, k}T_{s, k}+E_{p, k}+ T_c \sum_{m=1}^M u_{k,m} X_{k}(m) - E_k\right) \\
    +&\sum_{\ell^\prime=1}^L\sum_{\ell\neq\ell^\prime} \gamma_{\ell,\ell^\prime}\left(\alpha - \sum_{m=1}^{M} v_{\ell,\ell^\prime,m}\right) \\
    +& \sum_{k=1}^K\sum_{m=1}^M \theta_{k,m}\left[\frac{c_{k,m}^2}{u_{k,m}} - \hat{R}_{k,m}^{[t]}\left(\fb_m\right)\right] \\
    +&\sum_{m=1}^M\sum_{\ell^\prime=1}^L\sum_{\ell\neq\ell^\prime} \lambda_{\ell,\ell^\prime,m}\Big[Z_{m}\left(\{P_{s,k}\},\{c_{k,m}\},\fb_m\right) \\
    &\qquad\qquad\qquad\qquad\qquad \left. - \hat{Q}_{\ell,\ell^\prime,m}^{[t]}\left(\{c_{k,m}\},v_{\ell,\ell^\prime,m}\right)\right],
\end{aligned}
\end{equation}
where $\beta_k,\gamma_{\ell,\ell^\prime},\theta_{k,m}$ and $\lambda_{\ell,\ell^\prime,m}$ are all positive Lagrange multipliers. Then, some useful KKT conditions are given by
\begin{align}
    &\frac{\partial \Lcal_{\Prob{5}}}{\partial P_{s,k}} = \beta_k T_{s,k} - \frac{\sigma_r^2 c_{k,m}^2}{P_{s,k}^2} = 0,\\
    &\frac{\partial \Lcal_{\Prob{5}}}{\partial \fb_m} = 2N_0\fb_m - \sum_{k=1}^K 2\theta_{k,m}\hb_{k,m}\hb_{k,m}^H\fb_m^{[t]} = 0,\\
    &\frac{\partial \Lcal_{\Prob{5}}}{\partial c_{k,m}} = \frac{2\theta_{k,m}}{u_{k,m}} + \sum_{\ell^\prime=1}^L\sum_{\ell\neq\ell^\prime} \lambda_{\ell,\ell^\prime,m}\Bigg(\frac{2\sigma_r^2 c_{k,m}}{P_{s,k}}+ 2c_{k,m}\sigma_{s,k}^2 \notag\\
    &\quad\ \left.+2\sigma_m^2\sum_{k=1}^Kc_{k,m}-\frac{2\left(\sum_{k=1}^Kc_{k,m}^{[t]}\right)(\mu_{\ell,m}-\mu_{\ell^\prime,m})^2}{v_{\ell,\ell^\prime,m}^{[t]}}\right),
\end{align}
which can be respectively derived as below to reach the optimal value of $P_{s, k}^{[t+1]}$, $\fb_m^{[t+1]}$ and $c_{k,m}^{[t+1]}$
\begin{align}
    &P_{s,k}^{[t+1]} = \frac{\sigma_r c_{k,m}}{\sqrt{\beta_k T_{s,k}}}, \label{solve:P}\\
    &\fb_m^{[t+1]} = \frac{1}{2N_0}\sum_{k=1}^K 2\theta_{k,m}\hb_{k,m}\hb_{k,m}^H\fb_m^{[t]}, \label{solve:f}\\
    &c_{k,m}^{[t+1]} = \left[\displaystyle\sum_{\ell^\prime=1}^L\displaystyle\sum_{\ell\neq\ell^\prime}\lambda_{\ell,\ell^\prime,m}\left(\frac{\sum_{k=1}^Kc_{k,m}^{[t]}(\mu_{\ell,m}-\mu_{\ell^\prime,m})^2}{v_{\ell,\ell^\prime,m}^{[t]}}\right.\right. \notag\\
    &\qquad \left.-\sigma_m^2\displaystyle\sum_{k^\prime\neq k}c_{k^\prime,m}\Bigg)-\frac{2\theta_{k,m}}{u_{k,m}}\right]\Big/\Bigg[\left(\frac{\sigma_r^2}{P_{s,k}}+\sigma_{s,k}^2+\sigma_m^2\right) \notag\\
    &\qquad\quad \left.\left(\sum_{\ell^\prime=1}^L\sum_{\ell\neq\ell^\prime}\lambda_{\ell,\ell^\prime,m}\right)\right]. \label{solve:c}
\end{align}
Based on the results above, the multipliers $\beta_k,\gamma_{\ell,\ell^\prime},\theta_{k,m}$ and $\lambda_{\ell,\ell^\prime,m}$ can be updated with their stepsizes $\eta_{\beta_k},\eta_{\gamma_{\ell,\ell^\prime}},\eta_{\theta_{k,m}}$ and $\eta_{\lambda_{\ell,\ell^\prime,m}}$ to solve the problem in the next round, respectively. The primal-dual method is presented in Algorithm~\ref{alg:1}. Compared directly adopting the typical algorithms in existing toolbox like CVX, Algorithm~\ref{alg:1} enjoys the benefits of using the closed-form solutions in \eqref{solve:P}, \eqref{solve:f} and \eqref{solve:c}. Therefore, the computational complexity of Algorithm~\ref{alg:1} is reduced to $\Ocal(I(11N_r K M+N_r L^2 M+\frac{1}{2}L^2 K^2 M))$ with the assumption that Algorithm~\ref{alg:1} converges after $I$ loops of computing. 

\begin{algorithm}[!t]
    \renewcommand{\algorithmicrequire}{\textbf{Input:}}
    \renewcommand{\algorithmicensure}{\textbf{Output:}}
    \caption{Primal-dual method for solving $\Prob{5}$ in SCA iteration $t$}
    \label{alg:1}
    \begin{algorithmic}[1]
        \REQUIRE Channel gain $\{\hb_{k,m}\}$, feature elements' class centroid $\{\mu_{\ell,m}\}$ and variance $\{\sigma_{m}^2\}$, communication latency $\{T_{s,k}\}$ and other given parameters derived from iteration $t-1$
	\ENSURE $\{P_{s, k}^{[t+1]}, c_{k,m}^{[t+1]}, \fb_m^{[t+1]}, \alpha^{[t+1]}, u_{k,m}^{[t+1]}, v_{\ell,\ell^\prime,m}^{[t+1]}\}$
	\STATE Initialize $\{\beta_k\}$, $\{\gamma_{\ell,\ell^\prime}\}$, $\{\theta_{k,m}\}$, $\{\lambda_{\ell,\ell^\prime,m}\}$, the step size $\{\eta_{\beta_k}^{(0)}\}$, $\{\eta_{\gamma_{\ell,\ell^\prime}}^{(0)}\}$, $\{\eta_{\theta_{k,m}}^{(0)}\}$, $\{\eta_{\lambda_{\ell,\ell^\prime,m}}^{(0)}\}$ and $i \gets 0$;
	\WHILE {\textnormal{not convergence}}
            \STATE Derive $\{P_{s, k}^{[t+1]}\}$, $\{\fb_m^{[t+1]}\}$ and $\{c_{k,m}^{[t+1]}\}$ using \eqref{solve:P},\eqref{solve:f} and \eqref{solve:c}, respectively;
            \STATE Update the multipliers as 
            \begin{equation*}\left\{
            \begin{aligned}
                &\beta_k^{(i+1)} = \max\left\{\beta_k^{(i)}+\eta_{\beta_k}\frac{\partial \Lcal_{\Prob{5}}}{\partial \beta_k}, 0\right\},\\
                &\gamma_{\ell,\ell^\prime}^{(i+1)} = \max\left\{\gamma_{\ell,\ell^\prime}^{(i)}+\eta_{\gamma_{\ell,\ell^\prime}}\frac{\partial \Lcal_{\Prob{5}}}{\partial \gamma_{\ell,\ell^\prime}}, 0\right\},\\
                &\theta_{k,m}^{(i+1)} = \max\left\{\theta_{k,m}^{(i)}+\eta_{\theta_{k,m}}\frac{\partial \Lcal_{\Prob{5}}}{\partial \theta_{k,m}}, 0\right\},\\
                &\lambda_{\ell,\ell^\prime,m}^{(i+1)} = \max\left\{\lambda_{\ell,\ell^\prime,m}^{(i)}+\eta_{\lambda_{\ell,\ell^\prime,m}}\frac{\partial \Lcal_{\Prob{5}}}{\partial \lambda_{\ell,\ell^\prime,m}}, 0\right\};
            \end{aligned}\right.
            \end{equation*}
            \STATE $i \gets i + 1$;
	\ENDWHILE
    \end{algorithmic}
\end{algorithm}
As a result, the optimal solution of $\Prob{5}$ can be obtained and is denoted as $P_{s, k}^{[t+1]}$, $c_{k,m}^{[t+1]}$, $\fb_m^{[t+1]}$, $\alpha^{[t+1]}$, $u_{k,m}^{[t+1]}$, $v_{\ell,\ell^\prime,m}^{[t+1]}$, which are used as the reference points for the $(t+1)$-th iteration.

\subsubsection{Solution to $\Prob{3}$}
Based on the solution to $\Prob{5}$ and the SCA method described before, the solution procedure to $\Prob{3}$ is summarized in Algorithm~\ref{alg:2}.

\begin{algorithm}[!t]
    \renewcommand{\algorithmicrequire}{\textbf{Input:}}
    \renewcommand{\algorithmicensure}{\textbf{Output:}}
    \caption{Joint Sensing Power, Transmit Precoder and Receive Beamformer Design}
    \label{alg:2}
    \begin{algorithmic}[1]
        \REQUIRE Channel gain $\{\hb_{k,m}\}$, device energy $\{E_k\}$, sensing time $\{T_{s,k}\}$, communication time $T_c$, computation energy $\{E_{p,k}\}$
	\ENSURE $\{P_{s, k}^*\}$, $\{c_{k,m}^*\}$, $\{\fb_m^*\}$ and $\alpha^*$
	\STATE Initialize $t \gets 0$, $\{P_{s,k}^{[0]}\}$, $\{c_{k,m}^{[0]}\}$, $\{\fb_m^{[0]}\}$ in feasible region of $\Prob{4}$;
        \STATE Calculate the initial value of $\{v_{\ell,\ell^\prime,m}^{[0]}\}$;
        \STATE Initialize the auxiliary function $\hat{R}_{k,m}^{[0]}\left(\fb_m\right)$ and $\hat{Q}_{\ell,\ell^\prime,m}^{[0]}\left(\{c_{k,m}\},v_{\ell,\ell^\prime,m}\right)$;
	\WHILE {\textnormal{not convergence}}
            \STATE Derive $\Prob{5}$ by relaxing $\Prob{4}$ with $\hat{R}_{k,m}^{[t]}\left(\fb_m\right)$ and $\hat{Q}_{\ell,\ell^\prime,m}^{[t]}\left(\{c_{k,m}\},v_{\ell,\ell^\prime,m}\right)$;
            \STATE Solve $\Prob{5}$ with Algorithm~\ref{alg:1} to get optimum $\{P_{s, k}^{[t+1]}, c_{k,m}^{[t+1]}, \fb_m^{[t+1]}, \alpha^{[t+1]}, u_{k,m}^{[t+1]}, v_{\ell,\ell^\prime,m}^{[t+1]}\}$;
            \STATE Calculate the new auxiliary function $\hat{R}_{k,m}^{[t+1]}\left(\fb_m\right)$ and $\hat{Q}_{\ell,\ell^\prime,m}^{[t+1]}\left(\{c_{k,m}\},v_{\ell,\ell^\prime,m}\right)$;
            \STATE $t \gets t + 1$;
	\ENDWHILE
        \STATE The optimal solution $P_{s, k}^* \gets P_{s, k}^{[t]}$, $c_{k,m}^* \gets c_{k,m}^{[t]}$, $\fb_m^* \gets \fb_m^{[t]}$ and $\alpha^* \gets \alpha^{[t]}$;
    \end{algorithmic}
\end{algorithm}

\section{Simulation Results}\label{sec:5}

\subsection{Simulation Setup}
\subsubsection{Network settings}
A single-cell network is used to complete edge-device co-inference tasks. There is one edge server equipped with an $8$-antenna AP located at the center and $K$ single-antenna devices randomly located in a ring with radius in the range of [$R$, $R+0.05$] kilometers. By default, $R$ is set as $0.45$ and $K$ is set as $3$ unless specified otherwise. The channel gain of the link between the edge server and device $k$ is modeled as $\hb_k=|\varphi_k\rhob_k|^2$. $[\varphi_k]_{\dB}=-[\mathbf{PL}_k]_{\dB}+[\zeta_k]_{\dB}$ is the large-scale fading channel coefficient, where $[\mathbf{PL}_k]_{\dB}=128.1+37.6\log_{10}d_k$ is the path loss in dB, $d_k$ is the distance between device $k$ and the edge server, $[\zeta_k]_{\dB}$ is the shadowing in dB which follows the Gaussian distribution of $\Ncal(0,\sigma_\zeta^2)$. On the other hand, $\rhob_k\sim\CNcal(0,\Ib)$ stands for the small-scale fading channel coefficient, where Rayleigh small-scale fading is considered in the simulation. 
The variances of sensing noise $\sigma_r^2$ and clutter signal $\sigma_{s,k}^2$ are both set to $0.2$. The channel noise variance $N_0$ is set to $1$ and the variance of shadow fading $\sigma_\zeta^2=8$ dB.

\subsubsection{Inference tasks}
A human motion recognition task is selected to evaluate the performance of the proposed algorithm. The aim of this task is to distinguish $4$ human motions, i.e., adult pacing, adult walking, child pacing and child walking, where the heights of adults are uniformly randomized between $[1.6\textrm{m}, 1.9\textrm{m}]$ and the heights of children follow the uniform distribution in $[0.9\textrm{m}, 1.2\textrm{m}]$. The facing directions of adults and children are considered to be uniformly distributed in the range $[-180^\circ, 180^\circ]$ and the speed of moving is divided into three classes, with $0$ m/s, $0.25H$ m/s and $0.5H$ m/s representing standing, pacing and walking where $H$ is the height of each individual. The sensing time $T_s$ and communication time $T_c$ of devices are set to $1$ second and the computation energy $E_{p,k}$ is set to $0.1$ Joule. The dataset of radar sensing signals used for training and testing is generated by the wireless sensing simulator proposed in \cite{li2021wireless}.

\subsubsection{Inference Models}
To identify the motion from local features, two machine learning models are adopted: a support vector machine (SVM) model and a multi-layer perceptron (MLP) neural network. 
In this experiment, the MLP network is trained with Adam optimizer\cite{diederik2015adam}, with the numbers of neurons in the hidden layers of MLP set to $80$ and $40$. The dataset generated by the simulator proposed in \cite{li2021wireless} is separated into a training dataset containing 6400 samples and a test dataset which contains 1600 samples. The training dataset is considered as the ground-truth data (free of noise) to train both of the two ML models. The testing dataset is distorted by clutter and noise through the sensing process and communication process determined by the three schemes mentioned below. 

\subsubsection{Inference algorithms}
To verify the priority of the proposed scheme, three algorithms are compared in the experiments, as listed below. 
\begin{itemize}
    \item \emph{Our proposal}: All parameters are allocated by the proposed scheme in Algorithm~\ref{alg:2}.
    \item \emph{Existing AirComp scheme}: The sensing power is allocated randomly and other parameters are allocated following the AirComp scheme in \cite{wen2022task2}.
    \item \emph{Baseline}: The sensing power is allocated randomly, the receive beamforming is set to a constant of all elements' transmission and a maximum steering power is allocated under the energy constraint \eqref{cons:energy2}.
\end{itemize}

All experiments are implemented using Python 3.8.5 on a Windows 10 server with one NVIDIA\textsuperscript{\textregistered} GeForce\textsuperscript{\textregistered} GTX 1070 GPU 8GB and one Intel\textsuperscript{\textregistered} Core\texttrademark{} i7-8700 CPU.

\subsection{Performance Comparison}
In this part, the relations between the inference accuracy and the minimum pair-wise discriminant gain are firstly presented. Then, the impact of the cell radius on inference accuracy is analyzed. Finally, the three algorithms are compared in terms of the SVM model and the MLP model with different numbers of devices and different device energy thresholds, respectively.

\subsubsection{Relation between inference accuracy and minimum pair-wise discriminant gain}
The relations between the inference accuracy and the minimum pair-wise discriminant gain for both two machine learning models are illustrated in Fig. \ref{fig:acc_dg}. 
It shows that the inference accuracy grows from $25\%$ to $95\%$ as the minimum pair-wise discriminant gain increases for both AI models. Also, the SVM model reaches a higher inference accuracy than the MLP network, particularly, the accuracy of the SVM model gets nearly $40\%$ when the minimum pair-wise discriminant gain is $10$ while the accuracy of the MLP model is still $25\%$.

\begin{figure}[!t]
    \centering
    \includegraphics[width=\linewidth]{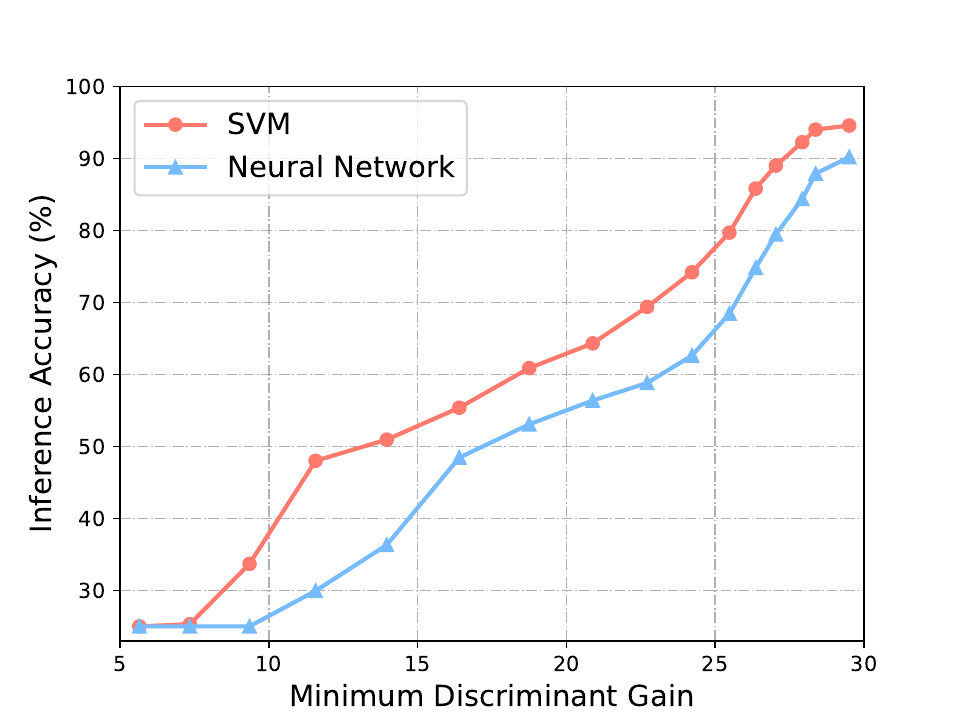}
    \caption{Inference accuracy versus minimum pair-wise discriminant gain on different models}
    \label{fig:acc_dg}
\end{figure}

\subsubsection{Relation between inference accuracy and cell radius}
Fig.~\ref{fig:acc_pl} presents the change of inference accuracy under different cell radiuses. It shows that the inference accuracies of both machine learning models decrease when the cell radius $R$ increases from $200m$ to $800m$. That's because the distances between the devices and the edge server turn to be larger with a larger $R$, leading to stronger path losses and weaker channel gains. This causes a larger communication distortion level and reduces the inference accuracy. Besides, Fig.~\ref{fig:acc_pl} also illustrates the effect of sensing distortion on inference accuracy. Both two machine learning models perform better in the case of low sensing distortion ($\sigma_r^2=0.2$) than in  the case of high sensing distortion ($\sigma_r^2=1.4$). 


\begin{figure}[!t]
    \centering
    \includegraphics[width=\linewidth]{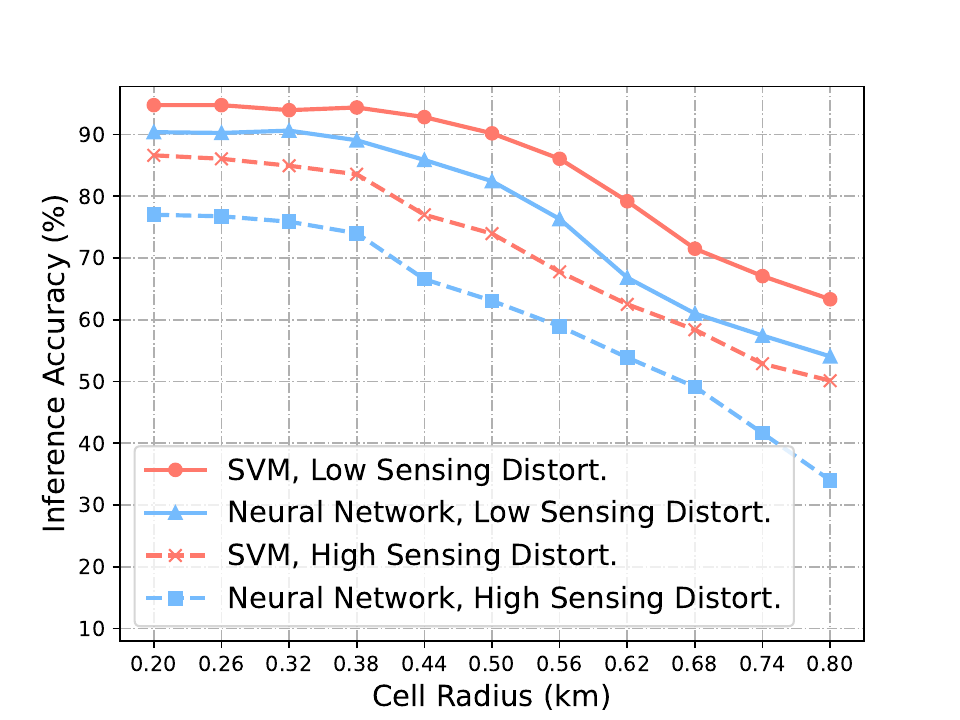}
    \caption{Inference accuracy versus cell radius on different models}
    \label{fig:acc_pl}
\end{figure}

\subsubsection{Inference accuracy v.s. number of devices}
The inference accuracies of the three schemes versus different number of devices are presented in Fig.~\ref{figgr:dev}. The performance of all three schemes increases as the number of devices increases for both machine learning models. It is because using more devices and aggregating their local features can reduce both the sensing distortion and communication noise. Besides, the proposed ISCC scheme outperforms the existing AirComp scheme proposed in \cite{wen2022task2}. The reasons are three folds. First, the proposed scheme adopts a more reasonable metric, say the minimum pair-wise discriminant gain, instead of the average pair-wise discriminant gain used in the existing AirComp scheme, leading to a balanced and enhanced achievable inference accuracy. Besides, the sensing stage of the inference task, which is separately designed in the existing AirComp scheme, is jointly designed in this work. Furthermore, rather than separately optimizing the aggregation of the feature elements in the existing scheme, they are jointly optimized, which allows more resources being assigned to the important elements. In addition, the inference accuracies of all scheme gradually saturate, since involving more devices has little contribution on suppressing the sensing and channel noise when the number of devices is large. The inference accuracy of the existing AirComp scheme saturates first because it's achievable inference accuracy is lower than that of the proposed scheme but it can well suppress the channel noise with a small number of devices.

\begin{figure*}[!t]
    \centering
    \subfloat[Inference accuracy of MLP versus the number of devices]{
        \includegraphics[width=.48\linewidth]{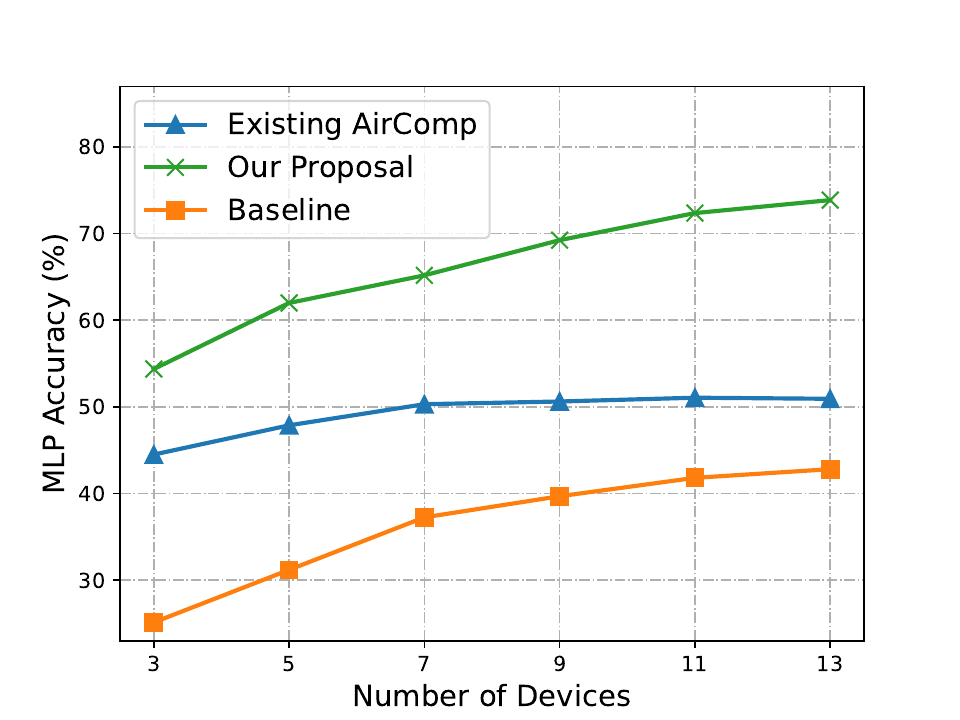}
        \label{fig:mlp_dev}
    }%
    \subfloat[Inference accuracy of SVM versus the number of devices]{
        \includegraphics[width=.48\linewidth]{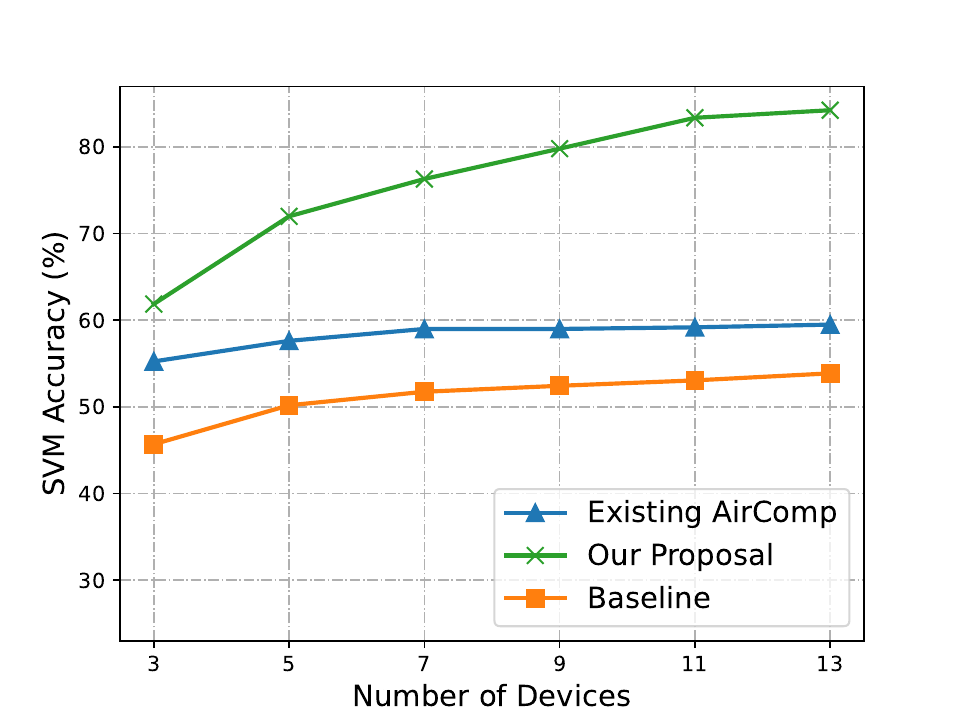}
        \label{fig:svm_dev}
    }
    \caption{Inference accuracy of different algorithms with different numbers of devices on MLP and SVM models}
    \label{figgr:dev}
\end{figure*}

\subsubsection{Inference accuracy v.s. device energy}
Fig.~\ref{figgr:pow} shows the impact of the device total energy on the accuracies of inference task in three schemes. It is shown that as a higher device energy is permitted, all of the three schemes have better inference accuracy since a higher device energy threshold means the devices can set larger sensing power and communication power to suppress the corresponding noise. In addition, the proposed scheme has a better performance than other two schemes for similar reasons as mentioned before.

\begin{figure*}[!t]
    \centering
    \subfloat[Inference accuracy of MLP versus device total energy]{
        \includegraphics[width=.48\linewidth]{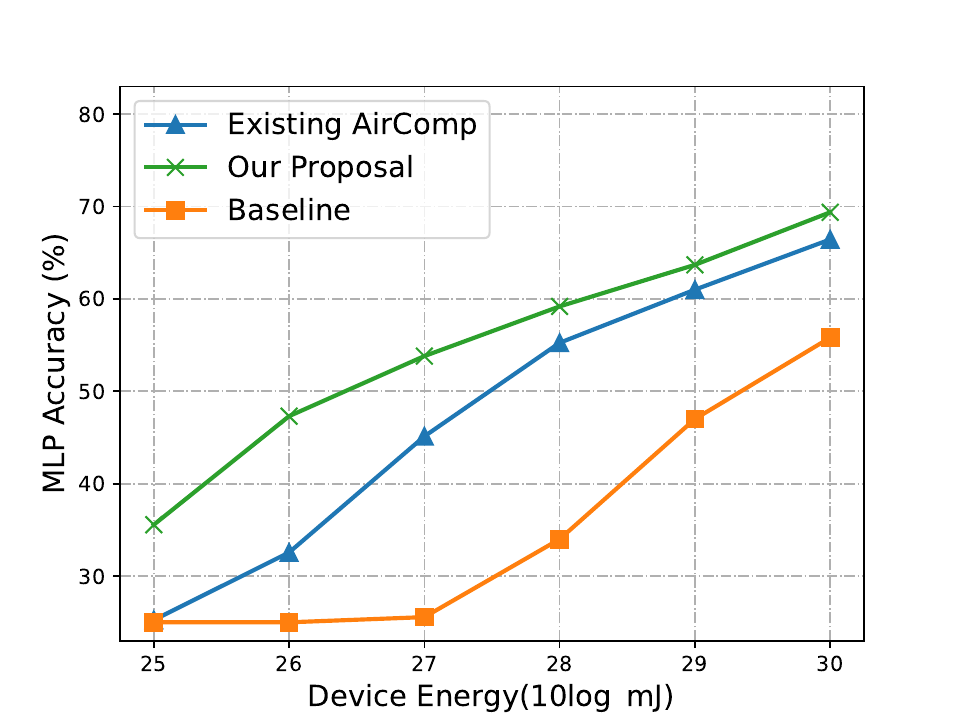}
        \label{fig:mlp_pow}
    }%
    \subfloat[Inference accuracy of SVM versus device total energy]{
        \includegraphics[width=.48\linewidth]{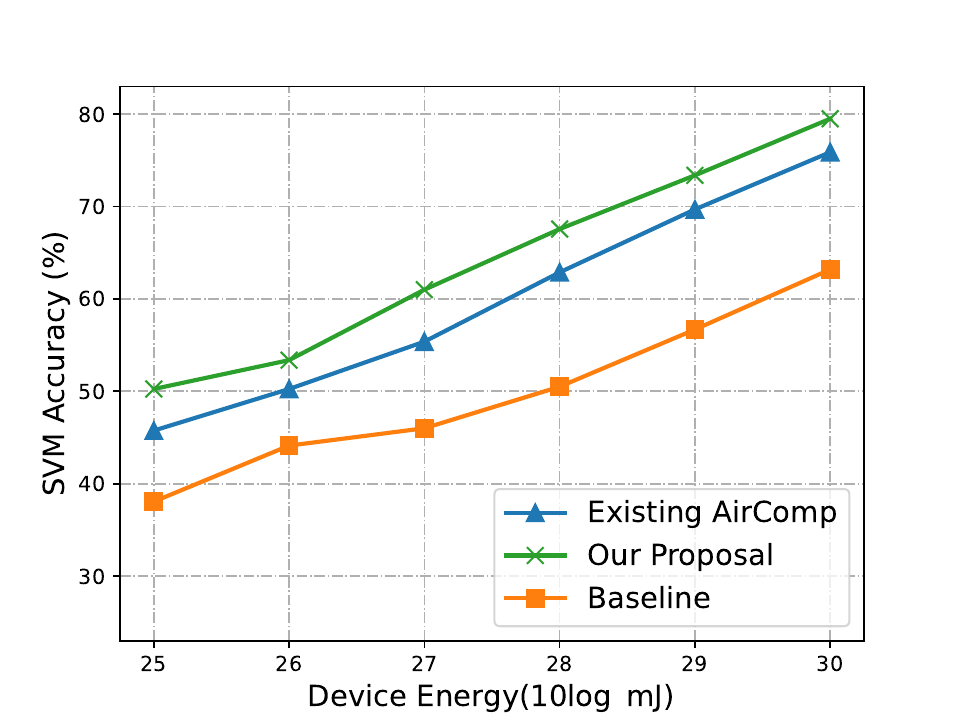}
        \label{fig:svm_pow}
    }
    \caption{Inference accuracy of different algorithms with different device total energy on MLP and SVM models}
    \label{figgr:pow}
\end{figure*}

\section{Conclusion}\label{sec:6}
In this paper, an AirComp based ISCC scheme was proposed for edge-device co-inference tasks. Compared to existing schemes, the proposed scheme enjoyed advantages from three aspects. To begin with, a novel design criterion, called maximum minimum pair-wise discriminant gain, was adopted, which enlarged the distance of the closest pair in the feature space, resulting in a balanced and enhanced achievable inference accuracy. Besides, the sensing, computation and communication processes were jointly investigated from a systematic view, allowing more flexible resource coordination and sharing among the three modules. Moreover, the aggregation of all feature elements was jointly designed, enabling adaptive resource allocation among different feature elements. Benefiting from the above three advantages, the proposed scheme enjoyed a more reasonable design goal and better resource utilization, thus leading to better inference performance compared to existing schemes as verified by the experiments.

This work opens several interesting directions for task-oriented ISCC scheme designs. One is to enhance the inference accuracy over time-variant channels or device scheduling under limited communication resources. Another is to design the scheme with some devices only acquiring part of the sensory view.

\appendix

\subsection{Proof of Lemma~\ref{lemma:1}}\label{appendix:1}
As mentioned in \eqref{def:GTx}, the ground-true feature element can be written as the average of $L$ independent Gaussian random variables
    \begin{equation}
        x(m) = \frac{1}{L}\sum_{\ell=1}^L x_\ell(m),
    \end{equation}
    where 
    \begin{equation}\label{dist:x}
        x_\ell(m)\sim\Ncal\left(\mu_{\ell,m},\sigma_m^2\right).
    \end{equation}
    Then by substituting it into \eqref{def:x}, the local feature element can be rewritten as
    \begin{equation}
    \begin{split}
        x_k(m)&=\frac{1}{L}\sum_{\ell=1}^L x_\ell(m)+c_{s, k}(m)+\frac{n_r(m)}{\sqrt{P_{s, k}}}\\
    &=\frac{1}{L}\sum_{\ell=1}^L x_{\ell,k}(m),
    \end{split}
    \end{equation}
    where $x_{\ell,k}(m)=x_\ell(m)+c_{s, k}(m)+n_r(m)/\sqrt{P_{s, k}}$. Thus, according to \eqref{dist:noise}, \eqref{dist:clutter} and \eqref{dist:x}, we can obtain the distribution of $x_{\ell,k}(m)$
    \begin{equation}
        x_{\ell,k}(m) \sim \Ncal\left(\mu_{\ell,m}, \sigma_m^2+\sigma_{s,k}^2+\frac{\sigma_r^2}{P_{s, k}}\right),~ \forall (k,\ell).
    \end{equation}
    Finally, the distribution of local feature element $m$ of device $k$ is given by
    \begin{equation}
        x_k(m) \sim \frac{1}{L}\sum_{\ell=1}^L\Ncal\left(\mu_{\ell,m}, \sigma_m^2+\sigma_{s,k}^2+\frac{\sigma_r^2}{P_{s, k}}\right),~ 1\leq k\leq K.
    \end{equation}

\subsection{Proof of Lemma~\ref{lemma:2}}\label{appendix:2}
Denote the optimal solution of the new problem $\Prob{3^\prime}$ as $\left\{ \{P_{s, k}^*\},\{c_{k,m}^*\},\{\fb_m^*\},\alpha^*\right\}$. Assume $\exists m^\prime\in[1,M]$ so that $\fb_{m^\prime}^*$ satisfy the following strict inequality:
\begin{equation}\label{l2:1}
    {c_{k,m^\prime}^*}^2 < {\hb_{k,m^\prime}^*}^H \fb_{m^\prime}^* {\fb_{m^\prime}^*}^H \hb_{k,m^\prime}^* u_{k,m^\prime}.
\end{equation}
Then, based on the continuity of quadratic function on the right-hand side of \eqref{l2:1} and for a fixed $c_{k,m^\prime}^*$, there always exists a number $\eta>0$ such that
\begin{equation}
    \fb_{m^{\prime-}}^* = (1-\eta)\fb_{m^\prime}^* \prec \fb_{m^\prime}^*,
\end{equation}
which leads to 
\begin{equation}
    \begin{aligned}
        {c_{k,m^\prime}^*}^2 &< {\hb_{k,m^\prime}^*}^H \fb_{m^{\prime-}}^* {\fb_{m^{\prime-}}^*}^H \hb_{k,m^\prime}^* u_{k,m^\prime}\\
        &< {\hb_{k,m^\prime}^*}^H \fb_{m^\prime}^* {\fb_{m^\prime}^*}^H \hb_{k,m^\prime}^* u_{k,m^\prime},
\end{aligned}
\end{equation}
where $\xb\prec\yb$ represents $\xb$ is element-wise less than $\yb$. By substituting $\fb_{m^{\prime-}}^*$ for the pair-wise discriminant gain constraint, the value of $\alpha$ can be increased to derive a better optimal value of $\Prob{3}$, which means that $\fb_{m^{\prime-}}^*$ is the optimal solution instead of $\fb_{m^\prime}^*$. However, this is a contradiction of the fact that $\fb_m^*$ is the optimal solution of $\Prob{3^\prime}$. Thus, the problem extended the constraint \eqref{cons:o1b} achieves the same optimal solution as $\Prob{3}$.

\subsection{Proof of Lemma~\ref{lemma:3}}\label{appendix:3}
It is quite apparent that the objective function, the first and second constraints of $\Prob{4}$ are all affine functions. Additionally, $R_{k,m}\left(\fb_m\right)$ is quadratic, which are convex and differentiable. Thus, we only need to prove that $c_{k,m}^2/u_{k,m}$, $Z_{m}\left(\{P_{s,k}\},\{c_{k,m}\},\fb_m\right)$ and $Q_{\ell,\ell^\prime,m}\left(\{c_{k,m}\},v_{\ell,\ell^\prime,m}\right)$ are convex and differentiable. 

Denote $f(x,y)=x^2/y$ with a positive $y$, we can derive the Hessian matrix
\begin{equation}
\Hb_f = 
\begin{bmatrix}
    \frac{2}{y} & -\frac{2x}{y^2} \\
     -\frac{2x}{y^2} & \frac{2x^2}{y^3}
\end{bmatrix},
\end{equation}
where the eigenvalues are 
\begin{equation*}
        \lambda_1 = 0, \quad \lambda_2 = \frac{2(x^2+y^2)}{y^3}.
\end{equation*}

Since $y>0$, both eigenvalues of $\Hb_f$ are non-negative, which indicates the Hessian matrix is positive semidefinite and thus $f(x,y)$ is convex.

By taking $x=c_{k,m}$ and $y=u_{k,m}$, it can be proved that $c_{k,m}^2/u_{k,m}$ is convex. Function $Z_{m}\left(\{P_{s,k}\},\{c_{k,m}\},\fb_m\right)$ is composed of three parts, the first part of which is the sum of $f(x,y)$ with $x=c_{k,m}$ and $y=P_{s,k}$ and the latter two parts are both quadratic. It follows that $Z_{m}\left(\{P_{s,k}\},\{c_{k,m}\},\fb_m\right)$ is convex and differentiable since linear transformation does not violate the convexity. Similar to $Z_{m}\left(\{P_{s,k}\},\{c_{k,m}\},\fb_m\right)$, function $Q_{\ell,\ell^\prime,m}\left(\{c_{k,m}\},v_{\ell,\ell^\prime,m}\right)$ can also be transformed from $f(x,y)$, which proves the convexity and differentiability. Thus, the third and fourth constraints are in the form of difference of convex functions and $\Prob{4}$ is a d.c. problem.

\bibliographystyle{IEEEtran}
\bibliography{refs}{}


 





\end{document}